\newif\iflong\longtrue
\newif\ifpodc\podcfalse

\iflong
\documentclass[11pt]{article}
\usepackage[top=1in,left=1in,right=1in,bottom=1in]{geometry}
\usepackage[utf8]{inputenc}
\usepackage{times,microtype}
\fi
\ifpodc
\documentclass[manuscript]{acmart}
\fi
\usepackage{amsthm}
\usepackage{amsmath,mathtools}
\usepackage{amsfonts}
\usepackage{hyperref}
\usepackage{enumitem}
\usepackage[noend]{algpseudocode}
\usepackage{algorithmicx}
\usepackage[boxruled,linesnumbered]{algorithm2e}
\usepackage[most]{tcolorbox}
\usepackage{color}
\usepackage{thm-restate}
\usepackage{soul}
\usepackage{hyperref}
\usepackage[capitalise,noabbrev,nameinlink]{cleveref}
\theoremstyle{plain}
\usepackage{mdframed}
\usepackage{hyphenat}
\usepackage{ragged2e}
\usepackage{array}

\newtheorem{theorem}{Theorem}[section]

\newtheorem{lemma}[theorem]{Lemma}
\newtheorem{corollary}[theorem]{Corollary}
\newtheorem{definition}[theorem]{Definition}

\newcommand{\epoch}[1]{\textsc{epoch}_{#1}}
\newcommand{\async}{T}

\newenvironment{informalthm}
  {\medskip\noindent{\textbf{Theorem (informal).}}}
  {\medskip}

\usepackage{array}
\usepackage{booktabs}

\iflong\else
\usepackage[font=small,aboveskip=0pt,belowskip=0pt]{caption}
\usepackage[compact]{titlesec}
\titlespacing*{\section}{0pt}{3pt}{3pt}
\titlespacing*{\subsection}{0pt}{3pt}{2pt}
\titlespacing*{\subsubsection}{0pt}{3pt}{2pt}
\fi

\newcommand{\floor}[1]{\lfloor{#1\rfloor}}

\newcommand{\mvdf}{\mathtt{VDF}_m}
\newcommand{\cvdf}{\mathtt{VDF}_v}
\newcommand{\vout}[1]{\rho_{#1}}

\newcommand{\vdf}[1]{\mathtt{VDF}_{#1}}

\newcommand{\proofout}[1]{\pi_{#1}}
\newcommand{\vrfout}[2]{v_{#1}^{#2}}
\newcommand{\vdfmout}[1]{M_{#1}}
\newcommand{\vdfcout}[1]{V_{#1}}
\newcommand{\vdfpout}[1]{P_{#1}}
\newcommand{\vdftout}[1]{C_{#1}}

\newcommand{\sk}[1]{\mathrm{sk}_{#1}}
\newcommand{\pk}[1]{\mathrm{pk}_{#1}}
\newcommand{\ek}[1]{\mathrm{ek}_{#1}}
\newcommand{\vk}[1]{\mathrm{vk}_{#1}}
\newcommand{\alg}{\mathcal{A}}
\newcommand{\prob}{\mathbb{P}}

\newcommand{\p}[2]{\pi_{#1}^{#2}}
\newcommand{\bloc}[2]{\block_{#1}^{#2}}
\newcommand{\block}{\mathcal{P}}

\newcommand{\diff}[1]{D_{#1}}
\newcommand{\nslow}{\delta_h'}
\newcommand{\nfast}{\delta_h}
\newcommand{\adv}{\delta_{adv}}

\newcommand{\ack}{\mathtt{ACK}}
\newcommand{\prop}{\mathtt{Propose}}
\newcommand{\tick}{\mathtt{Round}}
\newcommand{\pt}{D_0}

\newcommand{\env}{\mathcal{Z}}
\newcommand{\secp}{\kappa}
\newcommand{\secinp}{1^{\kappa}}
\newcommand{\adversary}{\mathcal{A}}

\newcommand{\eps}{\varepsilon}

\newcommand{\cryptosort}{\mathcal{F}_{mine}}

\newcommand{\setup}{\mathtt{Setup}}
\newcommand{\eval}{\mathtt{Eval}}

\newcommand{\veri}{\mathtt{Verify}}

\newcommand{\sX}{\mathcal{X}}
\newcommand{\sY}{\mathcal{Y}}

\newcommand{\qq}[1]{}
\newcommand{\nn}[1]{}
\newcommand{\tadge}[1]{}

\DeclareMathOperator{\poly}{poly}

\definecolor{mygreen}{RGB}{20,140,80}
\definecolor{linkcolor}{RGB}{0,0,230}
\definecolor{mylightgray}{RGB}{230,230,230}
\definecolor{verylightgray}{RGB}{245,245,245}

\hypersetup{
     colorlinks=true,
     citecolor= mygreen,
     linkcolor= mygreen,
     urlcolor= mygreen
}

\newcounter{myalgctr}

\newtcolorbox{OuterBox}[1][]{

    breakable,
    enhanced,
    frame hidden,
    interior hidden,
    left=-5pt,
    right=-5pt,
    top=-5pt,
    float=p,
    boxsep=0pt,
    arc=0pt
#1}

\newtcolorbox{InnerBox}[1][]{

    enforce breakable,
    enhanced,
    colback=gray,
    colframe=white,
#1}

\newenvironment{tbox}{
\vspace{0.2cm}
\begin{tcolorbox}[width=\textwidth,
                  enhanced,
                  boxsep=2pt,
                  left=1pt,
                  right=1pt,
                  top=4pt,
                  boxrule=1pt,
                  arc=0pt,
                  colback=white,
                  colframe=black,
	              breakable
                  ]

}{
\end{tcolorbox}
}

\newcommand{\tboxhrule}[0]{\vspace{0.1cm} {\color{black} \hrule} \vspace{0.2cm}}

\newenvironment{titledtbox}[1]{\begin{tbox}#1 \tboxhrule}{\end{tbox}}

\begin{document}
\title{A Lower Bound for Byzantine Agreement and Consensus for Adaptive Adversaries using VDFs}
\iflong
  \author{Thaddeus Dryja\thanks{Supported by the funders of the MIT Digital Currency Initiative.}\\
  MIT Media Lab\\
  \texttt{tdryja@media.mit.edu}
  \and
  Quanquan C. Liu\thanks{This research was conducted while working at and supported by the MIT Digital Currency Initiative.}\\
  MIT CSAIL\\
  \texttt{quanquan@mit.edu}
  \and
  Neha Narula\footnotemark[1]\\
  MIT Media Lab \\
  \texttt{narula@mit.edu}}
\fi

\iflong
\maketitle
\fi

\begin{abstract}

  Large scale cryptocurrencies require the participation of millions
  of participants and support economic activity of billions of
  dollars, which has led to new lines of work in binary Byzantine
  Agreement (BBA) and consensus. The new work aims to achieve
  communication-efficiency---given such a large $n$, not everyone can
  speak during the protocol. Several protocols have achieved consensus
  with communication-efficiency, even under an adaptive adversary, but
  they require additional strong
  assumptions---proof-of-work, memory-erasure, etc. All
  of these protocols use \emph{multicast}: every honest replica
  multicasts messages to all other replicas.
  Under this model, we provide a new communication-efficient consensus
  protocol using Verifiable Delay Functions (VDFs) that is secure
  against adaptive adversaries and does not require the same strong
  assumptions present in other protocols.

  A natural question is whether we can extend the synchronous
  protocols to the partially synchronous setting---in this work, we
  show that using multicast, we cannot. Furthermore, we cannot achieve
  \emph{always safe} communication-efficient
  protocols (that maintain safety with probability 1)
  even in the synchronous setting against a static adversary
  when honest replicas only choose to multicast its messages.
  Considering these impossibility results, we describe a new
  communication-efficient BBA protocol in a modified partially
  synchronous network model which is secure against adaptive
  adversaries with high probability.
  \footnote{Some versions of this paper have been distributed and circulated since October 2019.}
\end{abstract}

\ifpodc
\maketitle
\fi

\section{Introduction}

Consensus is a fundamental problem in distributed
systems. Historically, consensus protocols have been critical in
the context of ensuring the consistency of replicated data~\cite{chandra2007paxos, calder2011windows,
  borthakur2011apache}, but they were typically deployed with only a
few dozen replicas and only tolerated crash failures.
More recently, consensus protocols have been studied in the context of
cryptocurrencies to maintain a distributed public ledger. These
applications introduce new demands: First, cryptocurrency networks
operate with thousands or millions of participants (large $n$),
meaning $n^2$ communication complexity is unacceptable. Second, these
ledgers support billions of dollars of economic activity, so they need
to cope with a much stronger potential attacker.

Recent work addresses this goal of consensus with subquadratric
communication complexity while tolerating adaptive adversaries, but
these works require strong additional assumptions: Nakamoto's elegant longest
chain protocol~\cite{nakamoto} relies on idealized proof-of-work,
which has led to energy-intensive mining.  Algorand~\cite{algorand}
and Ouroborous Praos~\cite{ouroboros}
require honest users to erase their private keys from memory before
sending a message, known as the \emph{memory-erasure} model, which can be
difficult to ensure in practice. \cite{CPS19} uses a primitive called
\emph{batch agreement} which puts semantic requirements on agreement
values, meaning it is impractical to use in the cryptocurrency
context.
In light of these restrictions, we seek to answer the following question:

\begin{quote}
\emph{What communication-efficient consensus protocols secure against adaptive adversaries
can we obtain without strong cryptographic assumptions, and what are the limitations to obtaining these protocols?}
\end{quote}

Addressing this in even a synchronous network is challenging because
most known communication-efficient protocols use \emph{committee
  election}; proposals and voting are done by a leader and small
committee which are elected uniformly at
random. Typically the size is much smaller than the tolerated number
of faults, so an adaptive adversary can simply corrupt the leader and
entire committee, and vote for two values: we call this \emph{key
  reuse}. Memory-erasure is one technique to eliminate key reuse;
another is \emph{vote-specific eligibility} where election is dependent
probabilistically on the proposed value, so the adversary cannot force
a compromised leader and committee to vote for another value (with
high probability). Unfortunately, in these protocols the adversary can
use computational power to bias the elections: we call this \emph{vote
  grinding}. In the case of public ledgers, honest replicas will only
propose transactions sent to them by clients, which means that honest
replicas do not have a disproportionate chance to become part of the
committee. The adversary, on the other hand, can create and try many
different arbitrary transactions, for example spending coins back to
itself, to increase the chances that Byzantine replicas are elected to
committees.

Our solution is to make key reuse expensive by using Verifiable Delay
Functions (VDFs)~\cite{vdf} to make it temporally expensive to
send multiple votes. Leaders and committees are elected
in such a way that there is no opportunity for vote grinding. Note
that VDFs are proofs of \emph{sequential} computation, meaning
participants do not benefit from having parallel computational resources;
this is quite different from proofs-of-work.

Extending these protocols to operate in a partially synchronous
network (where a network is asynchronous
until an unknown, but finite \emph{global stabilization time}
is reached) introduces new challenges.
During a long enough asynchronous
period the adversary can drop messages and force repeated elections
until \emph{eventually a long sequence of leaders and committees are
elected where the adversary has an advantage}. We call this attack
\emph{fast-forwarding}. At this point the adversary can propose and
vote for multiple proposals, violating safety.
This is the reason why, even for our linear multicast
protocol in~\cref{sec:bba}, we require a bound on the number of rounds the asynchronous
period can last.
This, among other reasons, is the
motivation for developing a
new partially synchronous model
and a protocol in the modified partially
synchronous network in~\cref{sec:sublinear-bba}.

\subsection{Summary of Contributions}
In this section, we describe the main results we present in this paper.

\paragraph{Main Result 1: Limitations of Communication-Efficient Protocols using Multicast (\cref{sec:impossible-bba}).}
Thus far, most currently known
and implemented communication-efficient protocols
(e.g.\ \cite{cc19,CPS19,algorand}) have all honest replicas
communicate with other replicas in the network via \emph{multicasts} (or broadcasts). In other
words, every message that an honest replica sends is broadcasted to
all other replicas in the network.
First, we show that in such
protocols where honest replicas only multicast messages (and do not
perform point-to-point communication), it is impossible to achieve a
communication-efficient protocol even under static adversaries in the synchronous model
where safety is always guaranteed. We prove our result for binary Byzantine Agreement (BBA), which means it also holds for consensus.

\begin{informalthm}
It is impossible to formulate a communication-efficient protocol
for binary Byzantine agreement that always guarantees safety (safety
is guaranteed with probability $1$)
while tolerating even a static adversary in the synchronous network model,
when honest replicas multicast messages.
\end{informalthm}

Ideally we could maintain communication-efficiency
with high probability (or liveness in polylogarithmic rounds with high probability)
but always maintain safety with probability $1$. We hope our impossibility result
might motivate researchers to
investigate communication-efficient
protocols which do not require all honest nodes to multicast all messages.
An important open question is whether using point-to-point messages can lead to
communication-efficient protocols where safety is always maintained or whether this impossibility
result also extends to protocols which use point-to-point communication.

Second, we extend an impossibility result given in~\cite{cc19} to show that it is impossible
to formulate a communication-efficient binary Byzantine agreement protocol
that achieves agreement with high probability in the partially synchronous model
(with global stabilization times) as defined in~\cite{dls88},
even with synchronous processors (only message delays are asynchronous).

\begin{informalthm}
It is impossible to formulate a communication-efficient protocol for
binary Byzantine agreement in the partially synchronous model (even
when processors are synchronous) that achieves agreement with high
probability against an adaptive adversary when honest
replicas multicast messages.
\end{informalthm}

Thus, it seems fruitful to look for alternative models of partial
synchrony modeled after the GST model provided in~\cite{dls88}
to achieve communication-efficiency. We do so in our third main result.

\paragraph{Main Result 2: Consensus using VDFs (\cref{sec:general-consensus}).} We introduce a new randomized communication-efficient
consensus protocol based on Verifiable Delay Functions (VDFs) that is safe even against (weakly) adaptive adversaries
in the synchronous model.
This protocol does not require proof-of-work or the memory-erasure model and can withstand the case
when adversaries can arbitrarily choose the inputs of Byzantine nodes as well as the transactions
and proposals of each such node.

\begin{informalthm}
Suppose honest replicas can compute a VDF with difficulty $D$ in $\nfast D$ time
and Byzantine replicas can compute the same VDF in $\adv D$ time.
There exists a communication-efficient
consensus protocol for \emph{any positive constants} $\nfast$ and $\adv$
that reaches consensus in $O\left(\log n\right)$ rounds even in the presence of
adaptive adversaries in the synchronous model with overwhelming
probability (in the security parameter $\kappa$) and high probability in $n$
assuming the total number of replicas
is $\geq 3f +1$ where $f$ is the number of Byzantine replicas.
\end{informalthm}

Intuitively, VDFs guarantee that obtaining the output of the function
given an input requires some number $D$ of sequential steps (for a $D$ chosen
when the function is initialized) even when parallel processors are available. Verifying the
output of such a function only requires $O(\log D)$ steps.
Although VDFs require sequential computation, this amount of
computation is vastly less than the computation necessary to perform proofs-of-work
since the ability for adversaries to parallelize the work has been eliminated (so more hardware--up to reasonable sizes--does not
imply a bigger advantage). We use VDFs instead of the memory-erasure model assumptions
to protect against adaptive adversarial corruptions of important proposers and committees.
The adversary must compute a VDF in order to send more messages.
However, we must solve some number of challenges including when adversaries can potentially
have fast VDF solvers that take some constant fraction of the amount of time required by VDF solvers
held by honest replicas. A description of these challenges and their solutions are presented in \cref{sec:general-consensus}.

\paragraph{Main Result 3: Communication-Efficiency under Adaptive Adversaries in the Partially Synchronous
with Randomly Dropped Messages Model (\cref{sec:sublinear-bba}).}

Due to our impossibility results, it seems necessary to relax the assumptions
of the partially synchronous model slightly in order to obtain meaningful communication-efficient protocols for binary Byzantine agreement.
Thus, we formulate the \emph{partially synchronous with randomly dropped messages} network model where during the asynchronous
period, each message has probability $p$ of being dropped. Thus, the adversary no longer is able to selectively drop
messages during the asynchronous period. We show that in this model, we can have a communication-efficient
protocol using honest multicast that reaches agreement with high
probability.

\begin{informalthm}
There exists a communication-efficient protocol which reaches binary Byzantine agreement in $O\left(\poly \log n\right)$
rounds after GST with high probability in the partially synchronous with randomly dropped messages network model
under (weakly) adaptive adversaries.
\end{informalthm}

\section{Related Work}
\label{sec:related}
\newcolumntype{C}[1]{>{\centering\let\newline\\\arraybackslash\hspace{0pt}}m{#1}}

\subsection{Consensus Protocols and Adaptive Adversaries}
\begin{table}[h]
  \centering
    \begin{tabular}{|C{3cm}||C{4cm}|c|C{3.5cm}|}
    \hline
        Consensus Protocol & Network Model & Multicast Complexity & Assumptions \\
    \hline
    \hline
        Algorand \cite{algorand} &  Synchronous & $O(\poly(\log n))$ & \nohyphens{Memory-erasure}, PKI\\
        \hline
        Herding \cite{CPS19} & Synchronous &  $O(\poly(\log n))$ & \nohyphens{Filtering transactions by age}, PKI \\
        \hline
        Ouroboros \cite{ouroboros} & Semi-Synchronous & $O(\poly(\log(\secp)))$
        & \nohyphens{Memory-erasure}, PKI\\
        \hline
        Nakamoto \cite{nakamoto} &  \nohyphens{Synchronous}
        & $O(\poly(\log(\secp)))$ & \nohyphens{Proof of Work}\\
        \hline
        \cite{cc19v4} & Partially Synchronous (Fixed, but unknown $\Delta$) & $O(\poly(\log(n)))$ & BBA, PKI\\
        \hline
        \textbf{This work} & Synchronous & $O(\poly(\log n))$ & VDFs, PKI \\
        \hline
        \textbf{This work} & Partially Synchronous Randomly Dropped Messages & $O(\poly(\log n))$ & BBA, PKI \\
    \hline
\end{tabular}
    \caption{Comparison to related work on randomized, communication-efficient consensus protocols that tolerate adaptive adversaries.}

    \label{table:results}
\end{table}
Traditional consensus protocols~\cite{DS83,dls88} require all replicas to send messages to all other replicas,
resulting in $O(n^2)$ communication complexity in a network with $n$ replicas. Because they have such large communication
complexity, most of these protocols can be modified to account for adaptive adversaries.
Some~\cite{abraham2019synchronous,DS83,KK09} can even be shown to be secure for
strongly adaptive adversaries that can perform after-the-fact removal. However,
for our intended application to large-scale distributed systems such as decentralized cryptocurrencies, we
would like protocols with lower communication complexities.

Leader election-based consensus protocols~\cite{pbft,hotstuff} reduce communication complexity by
electing a single leader per round who aggregates votes. These protocols do not easily tolerate an adaptive adversary.
HotStuff~\cite{hotstuff}, using a 3-round pipelined protocol, uses
signature aggregation techniques to reduce authenticator complexity
(number of digital signatures or message
authentication codes sent in messages) to
$O(n)$. HotStuff also has the nice property of \emph{responsiveness};
it proceeds at actual network delay instead of worst case network delay.
We use HotStuff's clever 3-round protocol in both our synchronous
consensus protocol (\cref{sec:general-consensus}) and in our partially
synchronous clock synchronization protocols (\cref{sec:bba},\cref{sec:sublinear-bba}).
However, a straightforward application of HotStuff is not sufficient
to achieve subquadratic message complexity while tolerating adaptive
adversaries: an adversary could continually corrupt the leader for at least
$n/3$ rounds, forcing a quadratic number of messages before
finding an honest leader and reaching agreement.
A primary contribution of our work is showing how to prevent these types of attacks.

Other recent works have been able to lower the communication complexity by using additional
techniques.
The breakthrough work of King and Saia~\cite{king2011breaking} presented a binary Byzantine Agreement
protocol in the adaptive adversaries setting with communication complexity $O\left(n^{1.5}\right)$ with
the assumption of authenticated channels. As in Algorand and Micali-Vaikuntanathan
~\cite{algorand, micali2017optimal}, King and Saia~\cite{king2011breaking}
also assume that replicas can securely erase secrets from memory.
Other works like the sleepy model of consensus~\cite{PS17} and Ouroboros~\cite{ouroboros} also use the memory-erasure model.
As discussed in Canetti et al.\ \cite{CEGL08}, erasures are hard to perform in real
software.

The famous Nakamoto consensus
protocol~\cite{GKL15,nakamoto,PSS17,Ren19} achieves $O(n \poly
\log(\secp))$ communication complexity assuming perfect proof-of-work
in the synchronous model even under adaptive adversaries. This work
proposed what is known as the longest-chain strategy, which results in
eventual consensus.
More recent protocols~\cite{DPS19,ouroboros,KR18,KRDO17,PS17,Shi19} also follow
Nakamoto's longest-chain strategy but unlike Nakamoto consensus, they remove the
proof-of-work assumptions by using a permissioned setting with a public-key infrastructure.
In these protocols, a replica has some chance of being
elected as leader in each round. When a replica is elected as leader, it signs the
block extending the current longest chain. For such protocols
to exhibit both safety and liveness, some additional constraints have
to be imposed on the validity of the timestamps contained in the blockchain.
However, these works do not guarantee small turnover time for adaptive adversaries
if the memory-erasure model is \emph{not used}
regardless of whether the leader election is randomized~\cite{DPS19,KRDO17,PS17}
or deterministic~\cite{KR18,Shi19}. In fact, the number of rounds
to consensus could be near-linear
since the adaptive adversary can continuously corrupt
the small number of players who talk.

A key way to achieve communication-efficiency is electing a small
($O(\poly \log{n})$-sized) committee to run a step of the
protocol~\cite{algorand,cc19,CPS19,ouroboros,DPS19,HMW18}. This committee
is much smaller than the typical $n/2 - \epsilon$ or $n/3-\epsilon$
ideal number of corruptions to tolerate, and as such, the adversary
can compromise safety by corrupting the entire committee and voting
for two different values at the same time. Algorand gets around this
using memory-erasure; keys are ephemeral and thus not available to
vote for another value~\cite{algorand,micali2017optimal}.
\cite{cc19} tolerates an adaptive adversary for binary Byzantine
Agreement by leveraging the innovative idea of \emph{vote-specific
  eligibility}: by tying voting eligibility to the proposal, the
adversary cannot simply compromise the leader and elected committee
after they send a message and force them to vote for two values at the
same time. This is because most likely, the proposers and/or
committees for the two values will be different (or have very small
overlap); thus compromising one committee for one proposal does not
ensure committee membership for a different proposal. Though this
works for binary Byzantine Agreement, it does not extend to consensus
for general values because it introduces what we call \emph{vote
  grinding}: the adversary can try many different input values to
influence committee selection and create a biased committee, as noted
in~\cite{CPS19}.
In an updated version of their work, they provide a BBA protocol
that works in a partially synchronous network, however,
they use a different model where $\Delta$ is fixed but unknown,
while our lower bound is in the model where $\Delta$
only holds after a Global Stabilization Time (GST)~\cite{cc19v4}.

Chan, Pass, and Shi~\cite{CPS19} nicely build on ideas from both of
these works and achieve communication-efficient consensus with an
adaptive adversary using vote-specific eligibility and the novel idea of
\emph{batch agreement}: transactions, batched
together in a block proposal, are scored according to when the replica first saw
the transaction; older transactions score higher than new. The
adversary cannot try many different values to influence the committee
because honest participants will only vote for the highest-scoring block.
Unfortunately, it is unclear how this might
work in practice; many blockchains sort transactions by fees instead
of first-seen in order to rate limit and deter
spam~\cite{nakamoto,wood2014ethereum}. Straightforwardly sorting by transaction
fee instead of age in \cite{CPS19} would mean that an attacker could
continuously create many self-spending high-fee transactions and send
them to different honest replicas, making the honest replicas disagree on
the highest scoring block. Unlike what occurs with old transactions (at some point, everyone
agrees on the set of oldest transactions), adversaries
can keep on generating different higher-fee transactions, leading to indefinite disagreements.
Table~\ref{table:results} summarizes the differences between our work
and these other communication-efficient consensus protocols that
tolerate adaptive adversaries.

Other works~\cite{CCGZ19,GKRZ11,HZ10} have looked at adversaries whose corrupting powers
are delayed by a round but for Byzantine Broadcast, which is a different problem than what is considered
in this paper. They have focused on a simulation-based notion of adaptive security for Byzantine Broadcast,
where the concern is that the adversary should not be able to observe what the sender
wants to broadcast, and then adaptively corrupt the sender to flip the bit.
They use what is called the \emph{atomic message} model where after adaptively corrupting
a replica $i$ the adversary cannot erase the message $i$ already sent this round and also
must wait for at least one maximum network delay before the corrupt $i$ can start sending
corrupt messages.

\subsection{Lower Bounds for Binary Byzantine Agreement Protocols}
\begin{table}[h]
  \centering
    \begin{tabular}{|l||c|C{2.5cm}|C{2.5cm}|c|C{3cm}|}
    \hline
        Work & Type & Network Model & Adversary & Lower Bound & Even Assuming \\
    \hline
    \hline
    \cite{cc19} & any & any & Strongly Adaptive & $\Omega(f^2)$ & PKI \\
    \hline
    \cite{DR85} & \nohyphens{Deterministic} & any & Static or stronger & $\Omega(f^2)$ &  Authenticated Channels \\
    \hline
    \textbf{This work} & any & any (safety guaranteed with probability $1$)
    & Static or stronger & $\Omega(f^2)$ & PKI \\
    \hline
    \textbf{This work} & any & Partially Synchronous (GST) & Adaptive & $\Omega(f^2)$ & PKI \\
    \hline
\end{tabular}
\caption{Comparison to related work on BBA lower bounds.}
    \label{table:lb-results}
\end{table}
Previously, Abraham et al.\ \cite{cc19} have shown that (possibly
randomized) protocols that achieve subquadratic message complexity
cannot tolerate a strongly-adaptive adversary. The proof of their
lower bound is inspired by Dolev and Reischuk~\cite{DR85} who showed
that any \emph{deterministic} consensus protocol must incur
$\Omega(f^2)$ communication complexity when assuming authenticated channels.
Abraham et al.\ \cite{cc19} also show that without a PKI,
no protocol with $C(\secp, n)$ multicast complexity can achieve
consensus under $C(\secp, n)$ adaptive corruptions even in the
synchronous model, when assuming the existence of a random oracle or a common
reference string, and even in the memory-erasure
model.  Table~\ref{table:lb-results} compares these lower bounds to ours.
Some other works have achieved expected quadratic communication complexity under various
settings that are similar to adaptive adversarial settings~\cite{AMNRY19,AMS19}
in modified synchronous and asynchronous models.

\paragraph{Other Lower Bound Results}
Previously,~\cite{CMS89,KY84} showed that any randomized $r$-round protocol must fail with
probabilite at least $\left(c \cdot r\right)^{-r}$ for some constant $c$;
in particular, randomized agreement with sub-constant failure probability
cannot be achieved in strictly constant rounds. Attiya and Censor-Hillel~\cite{AC08}
extended the results of~\cite{CMS89,KY84} on guaranteed termination of randomized BA
protocols to the asynchronous setting, and provided a tight lower bound.
Much more recently, following a series of works looking at lower bounds on the expected number
of rounds necessary to achieve Byzantine agreement of randomized protocols,  Cohen et al.\
\cite{CHMOS19} show that BA protocols resilient against $n/3$ adaptive corruptions terminate
at the end of the first round with $o(1)$ probability among other results.

\subsection{Consensus with Verifiable Delay Functions}

Verifiable Delay Functions (VDFs) were first introduced in \cite{vdf}, with a
related precursor in \cite{lenstra2015random}.  Newer blockchain protocols use
VDFs in consensus protocols (with various other assumptions) as an
unbiasable source of randomness or as a source of timing to progress
rounds~\cite{azouvi2018betting,ethresearchbeacon,cohen2019chia}.
To the best of our knowledge, we are the first to use Verifiable Delay
Functions not as a source of randomness (leader and committee election
are independent of VDF output) but to bound the number of
messages an adversary can send, specifically with the purpose
of deterring adaptive corruptions.

\section{Model}

There are $n$ participants in the network and the public keys of all participants are
common knowledge. We only consider systems consisting of
$n\geq 3f + 1$ replicas where $f$ is the maximum number of
Byzantine replicas present in the system for the duration of the
protocol.

\paragraph{Network}
In this paper, we only consider protocols (in both our impossibility results and our
protocol formulations) where the honest replicas multicast their messages.
Consistent with the termininology given in~\cite{cc19} and~\cite{CPS19},
we use the term \emph{multicast} to indicate when a replica sends a message to all replicas in the network. Henceforth,
we talk about the communication complexity\footnote{Consistent with the terminology used in~\cite{cc19}, we
refer to \emph{communication complexity} as the total number of \emph{messages} sent in the network
by honest replicas. Unlike other commonly used notions of communication complexity, we are \emph{not}
referring to the total number of bits sent in the network.} in
terms of the \emph{multicast complexity} (i.e.\ the number of multicasts)\footnote{Note here
that we explicitly count only the number of multicasts as opposed to the total number of bits sent in all messages.
This is due to the fact that all messages sent by networks using a PKI require signatures of size $\poly(\secp)$
under standard cryptographic assumptions. Furthermore, it is difficult to standardize such a measure as
the number of bits of a message also depends on the size of the proposal/transaction/function/etc.}
as opposed to the point-to-point communication complexity as conventionally stated in the literature.
Honest replicas multicast \emph{all} messages, but Byzantine nodes may send point-to-point messages to anyone in the network.
This means our goal is to achieve sublinear multicast complexity, or subquadratic communication complexity.
Replicas communicate with each other in a network via authenticated
channels.  In \cref{sec:general-consensus},
we are operating in the synchronous network model;
the protocol proceeds in rounds and channels may exhibit communication delay which we
model as $\Delta$. Messages reach their intended recipient after up to $\Delta$ delay.
In \cref{sec:impossible-bba} and \cref{sec:bba}, we consider a partially synchronous network where
communication delay is unbounded until some Global Stabilization Time (GST) after which delay is bounded by $\Delta$.
\footnote{There are also several other partially synchronous models of consensus, which we do not consider in this
paper.}

\paragraph{Protocol Execution}
We assume as in~\cite{cc19,CPS19} that honest replicas interact with some environment $\env(\secinp)$ (where $\secp$ is the security parameter)
that sends them inputs at the beginning of every round, and honest replicas may send outputs to the environment
$\env$ at the end of every round. We assume that honest replicas attempt to reach consensus on one of the inputs
they received from $\env$ at the beginning of the protocol. Honest replicas follow the protocol when determining
their outputs/messages.

We assume that Byzantine replicas are controlled by some adversary $\adversary(\secinp)$ which reads each of their
inputs, received messages, and has accesss to their internal states. Then, $\adversary(\secinp)$ decides the
Byzantine replicas' outputs/messages. Crucially, the outputs/messages sent by Byzantine replicas could have no relation
to the inputs received by these replicas. Such replicas can output/send any number of arbitrary messages
independent of what they
receive from $\env$.

\paragraph{Adversary}
Throughout this paper,
we only consider adaptive adversaries, although one of our impossibility results
holds even for static adversaries.
While static adversaries can only corrupt up to $f$ replicas before the start of the protocol,
\emph{adaptive adversaries} are defined as adversaries $\adversary$ which can corrupt up to $f$ replicas
adaptively, at any point during the execution. When an adaptive adversary
corrupts a replica that was previously honest, it gains access to the replica's internal state
(including its private key), and, henceforth, $\adversary$ controls the corrupted replica.
A corrupted replica remains Byzantine for the remainder of the execution of the protocol.
$\adversary$ does
not have access to the internal states of the honest replicas. We assume that $\adversary$ also
has polynomially bounded parallel processing power and
cannot guess the secret keys of honest replicas with high probability.
\footnote{With high probability (whp) is defined in our paper to be
    probability at least $1 - \left(\frac{1}{n^c} + negl(\secp)\right)$
for all constants $c$.}
$\adversary$ can coordinate the Byzantine replicas,
and can read all messages sent through the network, but cannot erase or alter messages sent by
honest replicas.\footnote{In some previous literature (e.g.\ \cite{cc19}), this type of adaptive
adversary is referred to as a \emph{weakly adaptive adversary}.}

As in~\cite{cc19}, we define replicas which are honest at the current time to be \emph{so-far
honest}, and replicas which remain honest till the end of the protocol to be \emph{forever honest}.
We also assume that in the synchronous model, $\adversary$ can reorder the messages
received by any replica and can delay any message an arbitrary amount of time $\leq \Delta$.
In the partially synchronous model, we assume that $\adversary$ can selective choose
arbitrarily large delays for messages during the asynchronous phase and can drop or reorder any
number of messages during that phase. After GST, we assume $\adversary$ follows
the behaviours of a synchronous adversary.

\paragraph{Agreement Conditions}
In the adaptive adversary model, \emph{all forever-honest
replicas} must agree on exactly one input given to a forever-honest replica
by $\env$ at the beginning of the protocol, with high probability
with respect
to the number of nodes in the protocol $n$ and the security parameter $\secp$.
\footnote{We generally assume that $n$ is at least polynomial in $\secp$: $n = \Omega(poly(\secp))$.}
More specifically, a correct protocol in our paper
maintains the following two safety and liveness guarantees:

\begin{enumerate}
    \item Safety: No two honest replicas commit to two different values with high
        probability with respect to $n$ and $\secp$.
    \item Liveness: The protocol terminates in $O(\poly\log(n))$ rounds w.h.p.\ with respect
        to $n$ and $\secp$.
\end{enumerate}

Additional background on the network and adversarial  models, as well as a more detailed
explanation of the challenges facing protocol designers can be found in~\cref{app:additional}.

\section{Preliminaries}\label{sec:prelims}
The protocols and impossibility results discussed in this paper rely on two
main cryptographic primitives: verifiable random functions (VRFs)
and verifiable delay functions (VDFs). We assume standard cryptographic assumptions.
We first define the cryptographic primitives we need in this paper and then define the
various other notation we use throughout the paper.

\subsection{Cryptographic Primitives}

For all of our protocols, we assume that a trusted setup phase is first used to generate
a public-key infrastructure (PKI) where each replica $i \in [n]$ obtains a cryptographic sortition
public key/private key pair: $(\pk{i}, \sk{i})$ (such a key pair could be a verifiable random function (VRF)~\cite{vrf}
public key/private key pair).

For clarity we provide a simplified, informal definition of cryptographic sortition (which can
be implemented via VRFs) here; to see the full formal definition of VRFs~\cite{vrf}, please
refer to~\cref{app:vrfs}.

\paragraph{Cryptographic Sortition} \emph{Cryptographic sortition} ensures the following three
properties:
\begin{enumerate}
    \item Replica $i$ using its secret key $\sk{i}$ (and some
        public, common input) can determine whether they are part of the
        voting committee and produce some output.
    \item All other replicas can verify (but not produce with all but negligible probability in the security
        parameter $\secp$) replica $i$'s output
        using $\pk{i}$.
    \item Lastly, the output is unique and is indistinguishable from random with high probability.
\end{enumerate}

As in~\cite{cc19}, we use the notation $\cryptosort$ for replicas to use as an oracle for
determining whether they are eligible to vote in a committee.
$\cryptosort$ satisfies the properties of cryptographic sortition as stated above.
More specifically, $\cryptosort$ is parameterized by replica $i$'s secret key $\sk{i}$,
takes some input $x$, $\cryptosort(\sk{i}, x)$, and returns some output that is
generated uniformly at random via some coin flip with appropriate probability. Furthermore,
$\cryptosort(\sk{i}, x)$ can provide some verification to other replicas that
use only $\pk{i}$ and some additional information that is given as output from the function.
We let the output value and proof $\cryptosort(\sk{i}, x) \rightarrow (\vrfout{}{i}, \p{}{i})$ be
$\vrfout{}{i}$ and $\p{}{i}$, respectively. One possible instantiation of $\cryptosort$ is via
verifiable random functions. Please refer to~\cref{app:vrfs} for the full formal definition
of VRFs.

In our paper, we also make use of an additional cryptographic primitive called verifiable delay functions (VDFs)~\cite{vdf}.
A VDF is a function that guarantees with all but negligible probability in $\secp$ that computing the function
takes some $\diff{}$ sequential steps by some measure of difficulty $\diff{}$ of the function.
$\diff{}$ number of sequential steps is required even given polynomial number of parallel processors.
We present the full formal definition of VDFs in~\cref{app:vdf}. In this paper,
we let $\vdf{a}$ be a VDF with difficulty $\diff{a}$. In our exposition, we assume that the
evaluation and verification keys are implied and passed into the function so we do not
expressively pass in these as parameters into the function. $\vdf{a}$ takes as input some $x$
and outputs some output $S$, $\vdf{a}(x) \rightarrow S$, where $S$ includes both the value of the
output as well as the proof.

\iflong
\subsection{Other Notations and Definitions}

We make abundant use of the Chernoff bound in our paper.

\begin{definition}[Chernoff Bound]\label{def:chernoff}
    Let $Y_1, \dots, Y_m$ be $m$ independent random variables
    that take on values in $[0, 1]$ where $\mathbb{E}[Y_i] = p_i$
    and $\sum_{i = 1}^m p_i = P$. For any $\gamma \in (0, 1]$,
    the multiplicative \emph{Chernoff bound} gives
    \begin{align*}
        \prob\left[\sum_{i=1}^m Y_i > (1+ \gamma)P\right] < \exp\left(-\gamma^2 P/3\right)
    \end{align*}

    and
    \begin{align*}
        \prob\left[\sum_{i=1}^m Y_i < (1-\gamma)P\right] < \exp\left(-\gamma^2 P/2\right).
    \end{align*}
\end{definition}

We use the phrase ``with high probability'' many times throughout this paper. When we say
``with high probability'', we mean with high probability with respect to $n$ and with overwhelming
probability with respect to $\secp$; in other words, with probability at least $1 - \left(\frac{1}{n^c} + negl(\secp)\right)$
for all constants $c$. Throughout the paper, we assume $n = \Theta(\poly(\secp))$.
\fi

\section{Impossibility Results for BBA Using Sublinear Multicasts}\label{sec:impossible-bba}

In this section, we present two impossibility results
regarding BBA protocols with adaptive adversaries: First,
we show that it is impossible to \emph{always} achieve BBA
in even the synchronous network using a sublinear number of
multicasts (this implies it is also impossible in the partially synchronous model).
Then, we show that it is impossible to achieve BBA with high probability in a
partially synchronous network (in the GST model) in $o(n)$ multicasts.
Both of these results are under our definition of BBA in a network where honest replicas are only allowed to multicast messages.
We consider the specific \emph{binary Byzantine agreement} problem
that is defined in~\cite{cc19}.\footnote{This was also referred to in
  later works as \emph{multi-value agreement}~\cite{CPS19}.}  We
redefine the problem here for convenience:

\begin{definition}[Binary Byzantine Agreement Problem (BBA)]
Given a network with $n$ replicas, each replica $i$ receives an input bit $b_i \in \left\{0, 1\right\}$. The problem asks whether
all replicas can reach an agreement that satisfies the following properties with high probability:\footnote{High probability is generally defined to be
probability $1 - \frac{1}{n^c}$ for all constants $c$.}
\begin{enumerate}
\item \emph{Termination:} Every forever-honest replica $i$ outputs a bit $b_i$.
\item \emph{Consistency:} If two forever-honest replicas output $b_i$ and $b_j$, respectively, then $b_i = b_j$.
\item \emph{Validity:} If all forever-honest replicas receive the same input bit $b$, then all forever-honest replicas ouput $b$.
\end{enumerate}
\end{definition}

The proofs we present only apply to protocols where
all honest replicas multicast messages, meaning they, by their protocols, do not selectively choose
to send messages to a specific replica but instead multicast all messages to all replicas.
The Byzantine replicas are not constrained in this way and can send any number of point-to-point messages.
Our impossibility results apply to protocols with this assumption.
We define this property as the \emph{honest total multicast} property:

\begin{definition}[Honest Total Multicast Protocols]
Protocols where honest replicas multicast \emph{all} messages to all other replicas. Thus, the multicast complexity for
such protocols equals the number of times honest replicas multicast messages.
\end{definition}

\begin{lemma}\label{lem:sublinear-multicasts}
Any correct honest total multicast protocol in the synchronous model
with $M$ multicast complexity has at most $M$ honest replicas which multicast
before consensus is reached.
\end{lemma}

The proof of the aforementioned lemma immediately follows from the definition of honest total multicast protocols.

\begin{theorem}\label{thm:bba-total-impossible}
A honest total multicast protocol that uses sublinear $o(n)$ multicasts
with high probability and always reaches BBA in the
synchronous model cannot exist, even against a static adversary.
\end{theorem}

\begin{proof}
Supppose, for the sake of contradiction, that we have a correct honest total multicast BBA protocol that achieves sublinear $o(n)$ multicast complexity
with high probability and always reaches agreement on a bit.
Then, suppose that during one iteration of the protocol on a set of $n$ replicas, the protocol reaches agreement
wlog on the bit $1$. Such an iteration must exist since the protocol must reach agreement on $1$ if e.g.\ all inputs to all replicas is $1$.
Let this iteration of the protocol be $A$. Since the protocol guarantees agreement in sublinear multicast complexity with high probability,
we can also assume $A$ uses sublinear number of multicasts (as such an iteration must exist). Thus, there
exists some fraction of replicas which never multicast any messages in $A$. Let this set of replicas be $X_A$. Let the set of replicas that multicast
at least one message be $Y_A$. We know that $|Y_A| = o(n)$ by Lemma~\ref{lem:sublinear-multicasts}. Let $A$ reach agreement in $R_A$ synchronous rounds.

Suppose we have another iteration of the protocol on the same set of $n$ replicas, but where agreement is reached on $0$.
Again, such an iteration must exist since all honest replicas must output $0$ if e.g.\ all inputs to honest replicas are $0$.
Let this iteration of the protocol be $B$. Let the set of replicas which never multicast any messages be $X_B$ and the set
of replicas that multicast at least one message be $Y_B$. As before, we know that $|Y_B| = o(n)$ by Lemma~\ref{lem:sublinear-multicasts}.
Let $B$ reach agreement in $R_B$ synchronous rounds.

Suppose the adversary picks $\floor{n/3} - 1$ Byzantine replicas initially before the start of the protocol uniformly at random.
Let $\mathcal{S}_{A, B}$ be a simulation of the protocol on the set of $n$ replicas where all replicas in $Y_A \cup Y_B$ are initially
corrupted by the adversary. This is possible for large enough $n$ since $Y_A \cup Y_B = o(n)$. Furthermore, let
half of the replicas in $\left(X_A \cup X_B\right) \setminus \left(Y_A \cup Y_B\right)$ have input $1$ and have the same internal state as the same replicas in iteration $A$. Let this half be $H_1$.
Let the other half of the replicas in $\left(X_A \cup X_B\right) \setminus \left(Y_A \cup Y_B\right)$ have input $0$ and
have the same internal state as the same replicas in iteration $B$.
Let this half be $H_0$. Such a simulation is a potential iteration of the
protocol since before any messages are sent the internal states of all replicas are determined solely by their inputs
and their private random coin flips.

The adversary in simulation $\mathcal{S}_{A, B}$ then sends two sets of messages by controlling the replicas in $Y_A \cup Y_B$.
They send the same messages as in iteration $A$ to all replicas in $H_1$ and the same messages as in iteration $B$ to all replicas in $H_0$.
In this simulation, we assume all private coin flips for replicas in $H_1$ correspond with the same replicas in $A$ and all private
coin flips for replicas in $H_0$ correspond with the same replicas in $B$. Then, the replicas in $H_1$ have no way to distinguish
$\mathcal{S}_{A, B}$ from $A$ and will output $1$. Similarly, the replicas in $H_0$ have no way to distinguish $\mathcal{S}_{A, B}$
from $B$ and will output $0$.

Thus, we reach a contradiction as honest replicas $H_1$ agreed on $1$ and honest replicas $H_0$ agreed on $0$. Thus,
there does not exist a honest total multicast protocol that \emph{always} reaches BBA in the synchronous model,
even against a static adversary,
as there exists a potential simulation of the protocol that reaches agreement on two different bits.
\end{proof}

Our next impossibility result shows that there does not exist a
partially synchronous BBA protocol (in the GST model) with an
adaptive adversary that achieves agreement in $o(n)$ multicasts.
We need to be somewhat careful in our definition of multicast complexity in the partially synchronous model so that we obtain a definition that is
meaningful. What makes the partially synchronous model with
adaptive adversaries appealing is that it accurately
simulates the real world: dropped messages can be simulated by an
adversary which doesn't send messages (or selectively sends messages) to
different replicas. We define the \emph{multicast complexity} to be the total number of multicasts necessary after the global stabilization time (GST)
before Byzantine agreement is reached. In contrast to the synchronous model,
the asynchronous period starts at the beginning of the protocol and continues for unknown, but bounded time.
However, for the partially synchronous model, we assume the synchronous period after one GST must be long enough for the protocol to reach consensus.\footnote{In a system model where there can be
multiple synchronous periods separated by asynchronous periods and thus multiple GSTs, the synchronous period after a GST only needs to last long enough for one round of the protocol to complete.}

Our proof uses the lower bound proof given in Theorem 4 of~\cite{cc19}.

\begin{theorem}\label{lem:lower-partially-sync-bba}
There does not exist a partially synchronous
BBA protocol resilient against adaptive adversaries
where all honest replicas reach agreement with high probability
in $o(n)$ multicasts given $f$ Byzantine replicas
and for all $n \geq 3f + 1$ number of replicas in the network.
\end{theorem}

Because the proof of~\cref{lem:lower-partially-sync-bba} is very
similar to the proof of Theorem 4 of~\cite{cc19}, we relegate this
proof to~\cref{app:impossible-bba}.

We show in~\cref{sec:sublinear-bba} a BBA protocol that achieves agreement
with high probability in a new, weaker adverarial model than the
partially synchronous (GST) model.

\section{Consensus with Adaptive Adversaries using Sublinear Multicasts}\label{sec:general-consensus}

We use the concepts expanded upon
in the previous sections to formulate a communication-efficient
consensus protocol
without the use of the memory-erasure model
and which can be adapted to a variety of
transaction ordering schemes (e.g.\ for use
in cryptocurrency applications).
Namely, we make use of several important concepts in formulating our protocol:
verifiable delay functions (VDFs)~\cite{vdf}, random leader/committee elections,
and the three-step commit rule of HotStuff~\cite{hotstuff}.
The consensus protocol we describe in this section operates in the
synchronous model and can tolerate up to $\left(\frac{1}{3} - \eps\right)n$
adaptive Byzantine corruptions.

First, we provide a brief description and
a simplified version of our protocol in \cref{sec:brief-overview}. Then, we describe the full detailed version
of our protocol in \cref{sec:full-detailed-protocol}.
In our protocol, safety and liveness
hold with high probability with respect to $n$ and $\secp$
using $n\log^{O(1)}n$ number of messages
or $\log^{O(1)} n$ multicasts. The
exact multicast complexity, round complexity and
the proof of high probability
by which this holds provided in~\cref{thm:general-final}
is proven later in our
analysis in \cref{sec:consensus-analysis}.
As we showed in our lower bound result presented in \cref{sec:impossible-bba}, we
cannot guarantee that safety \emph{always} holds given a protocol that uses
sublinear multicasts even in the synchronous model and even given a static adversary.
Thus, our protocol ensures the best possible guarantees under the constraints we are operating under:
both safety and liveness with high probability with respect to $n$ and $\secp$.

\begin{theorem}\label{thm:general-final}
    Assuming a valid VDF construction that satisfies~\cref{def:vdfs},
    there exists a consensus protocol that terminates in $O(\log n)$ rounds
    and reaches consensus using $O(\poly \log n)$ multicasts with high probability with respect to $n$ and $\secp$,
    even when assuming the adversary can perform VDF computations faster by any constant factor $c > 0$.
\end{theorem}

\subsection{Protocol Overview}\label{sec:brief-overview}
In our protocol, we divide the communication rounds into \emph{epochs} where each
epoch goes through a leader election as well as several \emph{rounds} of communication to
confirm a leader's proposal. A leader is elected after each honest replica $i$ queries
$\cryptosort(\sk{i}, \ell)$ with its secret key and epoch number $\ell$ as input.
Recall from \cref{sec:prelims}
that each replica $i$ has oracle access to an oracle  $\cryptosort$ which will produce
some output and potentially a proof.
The leader $L$,\footnote{With high probability in $O(poly(\log(n)))$ rounds
, there will be one round where there is only one leader.}
then computes a VDF output of the value $L$ wants to propose. After computing
this VDF output, $L$ sends the VDF output, the proposal, the output of $\cryptosort(\sk{i}, \ell)$ and
proofs to all other replicas via a multicast.

After a proposal (with an attached VDF output and proof) is made by $L$,
some number of replicas are elected into committees to vote on the proposal.
We use a total of three uniformly at random chosen committees, similar to the
three-step commit rule of HotStuff~\cite{hotstuff}, to determine when a proposed value
is committed. However, unlike HotStuff, our committees are polylogarithmic in size
with respect to the number of participants in our consensus protocol.
As in previous works which use player-replaceability (e.g.\ \cite{algorand}),
each committee is chosen independently, likely with an entirely new set of participants.

To determine membership in a committee, each replica $i$ passes into $\cryptosort(\sk{i}, \cdots)$
as input the epoch number $\ell$
and a label for the committee it is attempting to participate in.
Each committee only votes for proposals proposed
in the current epoch; they will never vote for a proposal that was proposed in the
previous epoch or a future epoch.
After a committee member has been chosen to participate in a committee,
they must compute a VDF on their intended
vote; otherwise, honest replicas will not accept the vote without a corresponding VDF output.
When a replica multicasts its vote, it multicasts its vote
along with its $\cryptosort(\sk{i}, \cdots)$ output, the VDF output,
and all associated proofs.

To instantiate the VDFs we use in our protocol,
we can use a number of recent VDF constructions by~\cite{Wes19,Pietr19,tvdf}
(some of which do not need trusted setup).
They show constructions for VDFs that,
given a difficulty level $D$, can be computed in $D + O(1)$ time and verified in $O(\log(D))$ time
given a small number of processors. But such constructions also
guarantee that even given polynomially many parallel processors\footnote{For an arbitrary polynomial.},
computing the output must take at least
$D - \eps D$ parallel time for small $\eps$. The
formal definitions of such functions are given in the Preliminaries (\cref{sec:prelims}).

Although, theoretically, most VDF constructions
with the same difficulty must be computed
within some additive factor of one another,
our protocol can in fact handle any VDF instantiations (in practice)
where the speed of computation of the VDFs
differ by any constant multiplicative factor. This means that
our protocol is secure (w.h.p.) even when considering
adversaries which may have faster VDF computing potential
up to any positive constant multiplicative factor.

We now formally describe our protocol below.

\subsection{Detailed Protocol}\label{sec:full-detailed-protocol}

Our detailed consensus protocol shown in Fig.~\ref{fig:proposal-consensus}
is run by every honest replica $i$. $i$ maintains the private state $\ell$ which is the
current epoch that $i$ is on. Recall that we defined
an \emph{epoch} to be a period of time consisting
of many communication rounds in which voting for a particular proposal is done.
In our protocol detailed below, each epoch consists of $R$ communication rounds;
while the adversary can determine the order of messages that arrive to replicas
in our protocol, they cannot delay any message by more than $\Delta$ delay.

Note that in contrast to other works which uses a VDF to compute an
unpredictable source of randomness, we simply use the VDF to enforce that
the creation of a proposal or vote take some fixed amount of time.
In our protocol, leaders and committees are privately predictable---
a replica can predict for which values of $\ell$ it will be leader
or on a committee. As in~\cite{cc19,CPS19}, since we are operating
in a permissioned system (with $n$ replicas),
this does not affect the correctness of our protocol.

\cref{fig:visual-protocol} shows a simplified visual representation of our protocol.

\begin{figure}[h!]
    \centering
    \includegraphics[width=0.8\textwidth]{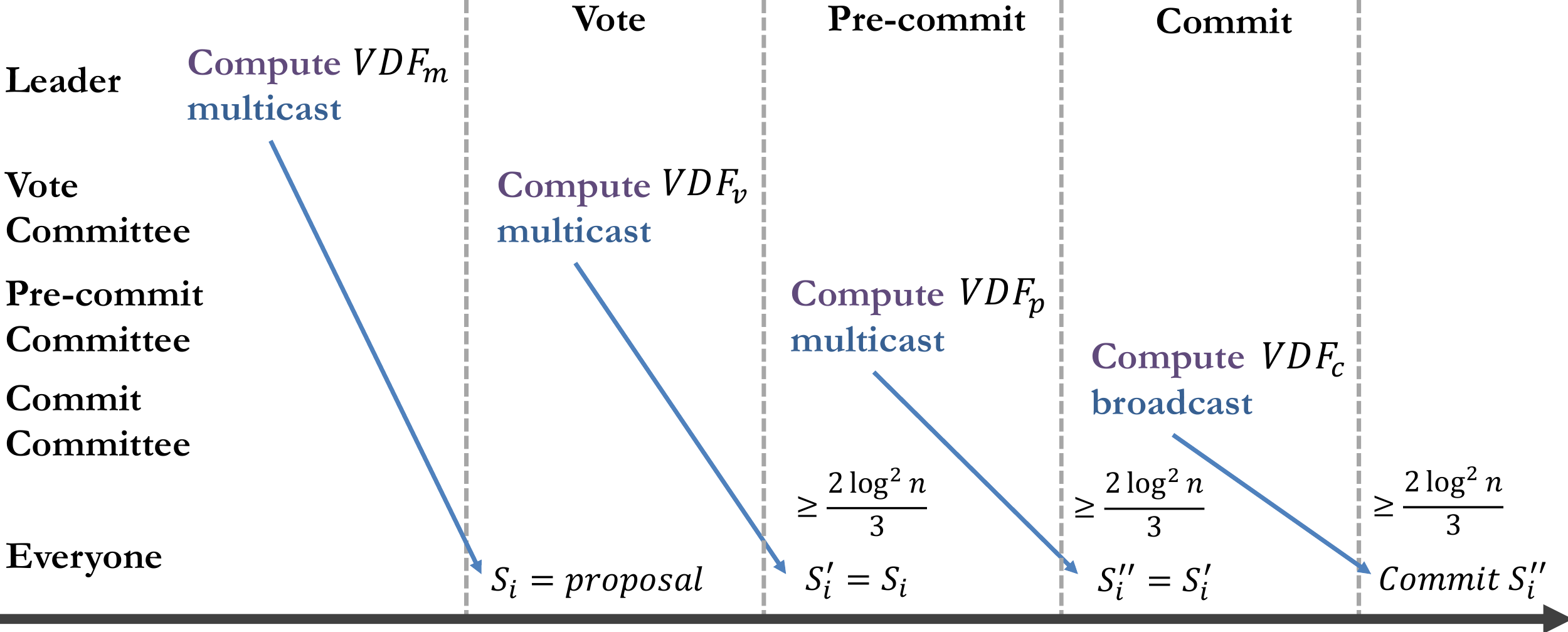}
    \caption{A simplified visual representation of our protocol detailed in~\cref{fig:proposal-consensus}.
    There are four total communication rounds--the leader proposal round, the voting round,
    the pre-commit round, and the commit round. In each round, the leader/committee members
    are chosen randomly via a call to $\cryptosort(\sk{i}, \cdots)$.}\label{fig:visual-protocol}
\end{figure}

\begin{figure}
\begin{mdframed}[hidealllines=false,backgroundcolor=gray!30]
\paragraph{Protocol for replica $i$:}
\begin{enumerate}
\item Let $\ell$ be the current epoch.
\item While a proposal has not been committed:
\begin{enumerate}
    \item Compute $\cryptosort(\sk{i}, \ell) \rightarrow (\vrfout{leader}{i}, \p{leader}{i})$
            where $\vrfout{leader}{i}$ is the output of the call to $\cryptosort$ and
            $\p{leader}{i}$ is the associated proof.
			\item If $\vrfout{leader}{i} \leq \frac{1}{2n}$, then:
			\begin{enumerate}
				\item Construct valid proposal $\bloc{\ell}{i}$.
				\item Compute $\mvdf(\ell, \vrfout{leader}{i}, \bloc{\ell}{i}) \rightarrow \vdfmout{i}$ which outputs $\vdfmout{i}$ that contains a VDF output value and a proof for the value.
				\item Once the output of $\mvdf$ has been computed, multicast $(\vdfmout{i}, \vrfout{leader}{i}, \p{leader}{i}, \bloc{\ell}{i})$.
				\item Set $S_i = \bloc{\ell}{i}$.
			\end{enumerate}
			\item Upon receiving a valid proposal $(\vdfmout{k}, \vrfout{leader}{k}, \p{leader}{k}, \bloc{\ell}{k})$ and $i$ is not a leader:
				\begin{enumerate}
                    \item Compute $\cryptosort(\sk{i}, \ell, \text{`}\mathtt{vote}\text{'}) \rightarrow (\vrfout{vote}{i}, \p{vote}{i})$.
					\item Compute $\cvdf(\ell, \vrfout{vote}{i}, \bloc{\ell}{k}) \rightarrow \vdfcout{i}$.
					\item If $\vrfout{vote}{i} \leq \frac{2\log^2 n}{3(1-\eps)n}$ (for some constant $\eps < 1/3$ defined in the analysis), then multicast
					$(\text{`}\mathtt{vote}\text{'}, \vrfout{vote}{i}, \p{vote}{i}, \vdfcout{i}, \bloc{\ell}{k})$.
					\item Set $S_i = \bloc{\ell}{k}$.
				\end{enumerate}
			\item Upon receiving $\geq \frac{2\log^2 n}{3}$ valid `$\mathtt{vote}$' messages for proposal $S_i$:
			\begin{enumerate}
                \item Compute $\cryptosort(\sk{i}, \ell, \text{`}\mathtt{precommit}\text{'}) \rightarrow (\vrfout{precommit}{i}, \p{precommit}{i})$.
				\item Compute $\vdf{p}(\ell, \vrfout{precommit}{i}, S_i) \rightarrow \vdfpout{i}$.
				\item If $\vrfout{precommit}{i} \leq \frac{2\log^2 n}{3(1-\eps)n}$, then multicast\\
				 $(\text{`}\mathtt{precommit}\text{'}, \vrfout{precommit}{i}, \p{precommit}{i}, \vdfpout{i}, S_i)$.
				\item Set $S_i' = S_i$.
			\end{enumerate}
			\item Upon receiving $\geq \frac{2\log^2 n}{3}$ valid `$\mathtt{precommit}$' messages for proposal $S'_i$:
			\begin{enumerate}
                \item Compute $\cryptosort(\sk{i}, \ell, \text{`}\mathtt{commit}\text{'}) \rightarrow (\vrfout{commit}{i}, \p{commit}{i})$.
				\item Compute $\vdf{c}(\ell, \vrfout{commit}{i}, S'_i) \rightarrow \vdftout{i}$.
				\item If $\vrfout{commit}{i} \leq \frac{2\log^2 n}{3(1-\eps)n}$, then multicast
				 $(\text{`}\mathtt{commit}\text{'}, \vrfout{commit}{i}, \p{commit}{i}, \vdftout{i}, S'_i)$.
				\item Set $S_i'' = S'_i$.
			\end{enumerate}
			\item Upon receiving $\geq \frac{2\log^2 n}{3}$ valid `$\mathtt{commit}$' messages for proposal $S''_i$, commit to $S_i''$ and set $S_i, S_i', S_i'' = \emptyset$.
			\item Timeout if none of the above steps can be taken after $\Delta$ time.
		\end{enumerate}
		\item After $R$ communication rounds (for $R$ defined in the analysis to be $X / \Delta$) and/or timeouts, terminate the while loop for epoch $\ell$, set $S_i, S_i', S_i'' = \emptyset$ and proceed with epoch $\ell+ 1$ of the protocol.
		\end{enumerate}
\end{mdframed}
\caption{General consensus protocol for honest replica $i$.}\label{fig:proposal-consensus}
\end{figure}

\subsection{Analysis}\label{sec:consensus-analysis}

As before, we define the following terms,
a \emph{round} of a replica $i$ consists of sending and/or receiving a set of messages
(in other words, one round of communication)
and an \emph{epoch} $\ell$ is defined
to be one iteration of the while loop defined in the protocol given in Fig.~\ref{fig:proposal-consensus}.
Assuming $\Delta$ message delay, we first prove that if there is exactly one leader--which is honest,
there are $\geq \frac{2\log^2{n}}{3}$ honest replicas in each committee,
and there are $< \frac{\log^2 n}{3}$
Byzantine replicas in each committee, then we can reach consensus on the
leader's proposal given appropriate initial settings of the parameters.

Let $\nslow$, $\nfast$ and $\adv$ be constants $\nslow, \nfast, \adv > 0$. We assume that the slowest honest replica
takes $\nslow D$ time to compute a VDF of difficulty $D$, the fastest honest replica takes $\nfast D$ time to compute the VDF, and
any Byzantine replica takes $\adv D$ time to compute the VDF.
We show that our protocol accounts for the most interesting settings of the parameters: $\nslow, \nfast > \adv$;
in the case when the adversary computes the VDF slower than honest replicas,
security can be proven trivially.
Let $X$ be the total time (in terms of $\Delta$) that each epoch consists of and $R = X/\Delta$
is the corresponding number of communication rounds.\footnote{In the case when $X$ is not
divisible by $\Delta$, we can increase the duration of $X$ such that it becomes divisible by $R$.}
We give the exact bounds for these variables, $X$ and $R$,
in our proofs (in terms of $\Delta$). Throughout our proofs, we let $D_x$ be the difficulty level of $\vdf{x}$.

\begin{lemma}\label{lem:parameter-setting}
Let $\Delta$ be the message delay. For epoch $\ell$,
suppose that there is exactly one leader, which is honest,
there are $\geq \frac{2\log^2{n}}{3}$ honest replicas in each committee,
and there are $<\frac{\log^2{n}}{3}$ Byzantine replicas in each committee.
When $\nslow, \nfast > \adv>0$ and $\nfast \approx \nslow$,
there exist values $\diff{m}, \diff{c}, \diff{p}, \diff{v}, X > \Delta$
in terms of $\nfast$, $\nslow$, $\adv$, and $\Delta$
that allow for the leader's proposal to be committed by all honest replicas
with high probability with respect to $n$ and $\secp$.
\end{lemma}

\begin{proof}
In the case where there is exactly one leader, who is honest, and all committees have
$\geq \frac{2\log^2{n}}{3}$ honest replicas and $<\frac{\log^2{n}}{3}$ Byzantine replicas,
each leader and honest committee member will send out exactly one proposal/vote. However,
the adversary can potentially choose to adaptively corrupt the leader and/or committee members
and send out multiple proposals if the difficulty levels of our VDFs are not set appropriately. The only
way that an adversary can send multiple proposals or votes is if they compute the VDFs associated with
the proposals or votes. Since we assume that the adversary cannot guess the private keys
of the honest replicas with all but negligible probability in $\secp$,
they cannot compute the VDFs of the extra proposals and votes until
after they corrupt the replicas with all but negligible probability in $\secp$
by~\cref{def:vdfs}.
By the assumptions given in the lemma statement, initially both
the leader and majority of committee members are so-far honest.
Thus, we need only concern ourselves
with the cases when the honest replicas are
corrupted after they announce their leadership/committee status.

In order to prevent the leader from sending multiple proposals, the leader must not have enough
time to compute a new VDF output on a new proposal
after computing the current VDF output on the proposal they have already
multicasted. Recall that by our definitions of $\nfast$ and $\adv$, the fastest that
an honest replica can compute $\mvdf$ is $\nfast \diff{m}$ and any
Byzantine replica must take
at least $\adv \diff{m}$ time to compute $\mvdf$.

We must ensure that each time a replica computes a VDF and sends the result, the adversary does not
have enough time to compute another value for the VDF before we proceed with the next epoch. Thus, the
difficulty levels of the VDFs must be set accordingly. Let $X$ be the time that an epoch lasts (in terms of $\Delta$)
before we proceed to the next epoch. Then, for example, for the leader
proposal round, the amount of time it takes for the
fastest honest replica to compute the corresponding proposal VDF plus the time it takes for the adversary to take
control of the honest proposer and compute another VDF must be longer than the length of the epoch. The
constraint on the difficulty level $\diff{m}$ must then follow: $(\nfast + \adv)\diff{m} > X$. Following this
pattern, the remaining difficulty terms must follow similar constraints.
Intuitively, this also means that $\diff{m} > \diff{v} > \diff{p} > \diff{c}$. Finally, $X$ must be long enough
so that honest replicas can compute, receive, and verify all VDF outputs so they can commit a proposal if the conditions of the
lemma are followed.

From the intuition above, the difficulty levels that are set must specifically follow the following constraints:

\begin{align}
(\nfast + \adv)\diff{m} > X \\
\nfast(\diff{m} + \diff{v}) + \adv \diff{v} > X \\
\nfast(\diff{m} + \diff{v} + \diff{p}) + \adv \diff{p} > X \\
\nfast(\diff{m} + \diff{v} + \diff{p} + \diff{c}) + \adv \diff{c} > X \\
X \geq \nslow (\diff{m} + \log(\diff{m}) + \diff{v} + \log(\diff{v}) + \diff{p} + \log(\diff{p}) + \diff{c} + \log(\diff{c})) + 4\Delta \label{eq:large}
\end{align}

We solve this set of equations to obtain the following set of expressions for
$\diff{m}$, $\diff{v}$, $\diff{p}$, and $\diff{c}$ in terms of $X$:

\begin{align}
\diff{m} > \frac{X}{\nfast+ \adv}\\
\diff{v} > \frac{X}{\nfast + \adv} - \nfast\left(\frac{X}{\left(\nfast + \adv\right)^2}\right)\\
\diff{p}>\frac{X}{\nfast + \adv} - 2\nfast\left(\frac{X}{\left(\nfast + \adv\right)^2}\right) + \nfast^2 \left(\frac{X}{\left(\nfast + \adv\right)^3}\right)\\
\diff{t} > \frac{X}{\nfast + \adv} - 3\nfast\left(\frac{X}{\left(\nfast + \adv\right)^2}\right) + 3\nfast^2 \left(\frac{X}{\left(\nfast + \adv\right)^3}\right) - \nfast^3 \left(\frac{X}{\left(\nfast + \adv\right)^4}\right)
\end{align}

Substituting the above into Eq.~\ref{eq:large} gives us a lower bound for $X$ from which we can also derive the other values. First
we replace $\log(\diff{})$ with $\eps' \diff{}$ for some small constant $\eps'$ for all $\diff{}$.

\begin{align}
X > \nslow\left((1+\eps')\diff{m} + (1+\eps')\diff{v} + (1+\eps')\diff{p} + (1+\eps')\diff{c}\right) + 4\Delta\\
X > \nslow(1+\eps')\left(\left(\frac{4X}{\nfast+ \adv}\right) - 6\nfast\left(\frac{X}{\left(\nfast + \adv\right)^2}\right) + 4\nfast^2 \left(\frac{X}{\left(\nfast + \adv\right)^3}\right) - \nfast^3 \left(\frac{X}{\left(\nfast + \adv\right)^4}\right)\right) + 4\Delta\\
X > \frac{4\Delta}{1 - \nslow(1+\eps')\left(\frac{4}{\nfast+ \adv}-\frac{6\nfast}{\left(\nfast + \adv\right)^2} + \frac{4\nfast^2}{\left(\nfast + \adv\right)^3}-\frac{\nfast^3}{\left(\nfast + \adv\right)^4}\right)}.
\end{align}

Substituting the expression for $X$ will lead to values of $\diff{m}$, $\diff{v}$, $\diff{p}$ and $\diff{c}$ in terms of the values of $\nfast$, $\nslow$, and $\adv$.

This expression is valid iff

\begin{align}
1 - \nslow(1+\eps')\left(\frac{4}{\nfast+ \adv}-\frac{6\nfast}{\left(\nfast + \adv\right)^2} + \frac{4\nfast^2}{\left(\nfast + \adv\right)^3}-\frac{\nfast^3}{\left(\nfast + \adv\right)^4}\right) > 0\\
\nslow < \frac{1}{(1+\eps')\left(\frac{4}{\nfast+ \adv}-\frac{6\nfast}{\left(\nfast + \adv\right)^2} + \frac{4\nfast^2}{\left(\nfast + \adv\right)^3}-\frac{\nfast^3}{\left(\nfast + \adv\right)^4}\right)}\label{eq:constraint-1}
\end{align}

and

\begin{align}
\frac{4}{\nfast+ \adv}-\frac{6\nfast}{\left(\nfast + \adv\right)^2} + \frac{4\nfast^2}{\left(\nfast + \adv\right)^3}-\frac{\nfast^3}{\left(\nfast + \adv\right)^4} > 0\label{eq:constraint-2}.
\end{align}

Eq.~\ref{eq:constraint-2} is always true for all $\nfast > 0, \adv > 0$. Hence, we need only concern ourselves with the constraint defined by Eq.~\ref{eq:constraint-1}. Assuming that $\eps'$ is negligible\footnote{Given that we pick $\eps'$ such that $\eps'D \geq \log(D)$, if $\eps'$ is not negligible, then we can increase the delay $4\Delta$ in Eq.~\ref{eq:large} to something greater to account for the time necessary to verify the VDF computations.}, we can simplify to obtain:

\begin{align}
\nslow < \frac{(\nfast + \adv)^4}{\nfast^3 + 4\nfast^2 \adv + 6\nfast \adv^2 + 4 \adv^3\label{eq:bounds}}.
\end{align}

For all values of $\nfast$, $\adv$, we obtain a bound for $\nslow$ where there exist values we can set $\nslow$ such that $\nslow > \nfast$. We have thus proven that
there exist values of $\diff{m}, \diff{c}, \diff{p}, \diff{v}, X > \Delta$ given $\nslow, \nfast > \adv > 0$
and $\nfast \approx \nslow$ that we can set
to prevent violation of safety by the corruption of so-far honest replicas.

In such cases, when the conditions given in the statement of the lemma are followed,
given exactly one honest proposer and committees dominated by honest replicas, the adversary is not able to
produce additional proposals or votes with all but negligible probability in $\secp$.
Furthermore, the adversary does not have enough time to corrupt an honest replica
and compute the associated message or vote VDF before the epoch has progressed to the next epoch.

Since a single honest leader will always propose exactly one proposal,
all honest replicas will vote for the same proposal, reaching the necessary number of votes.
Hence, the leader's proposal will be committed by all honest replicas.
\end{proof}

Now, we remove the constraint of $\nslow \approx \nfast$ by assuming that each honest replica with a faster
VDF implementation than $\nslow$ can choose to delay sending their proposal or vote until after the time that it would have
taken the replicas that take $\nslow \diff{}$ time to compute and verify the VDFs. This immediately allows us to conclude that our
protocol can handle any constant values of $\adv, \nfast > 0$ (since the constraint $\nfast \approx \nslow$ in Lemma~\ref{lem:parameter-setting}
is trivially satisfied). For Corollary~\ref{cor:slow-down}, we assume that all honest replicas compute the VDFs with speed
$\nfast$.

\begin{corollary}\label{cor:slow-down}
Let $\Delta$ be the message delay. For epoch $\ell$, suppose that there is exactly one leader, which is honest,
there are $\geq \frac{2\log^2{n}}{3}$ honest replicas in each committee,
and there are $<\frac{\log^2{n}}{3}$ Byzantine replicas in each committee.
When $\nfast > \adv > 0$,
there exist values $\diff{m}, \diff{c}, \diff{p}, \diff{v}, X > \Delta$
in terms of $\nfast$, $\adv$ and $\Delta$
that allow for the leader's proposal to be committed by all honest replicas.
\end{corollary}

In the rest of this section, we prove the safety and liveness of
our consensus protocol which directly leads to the proof of~\cref{thm:general-final}.

We first show that each epoch consists of a constant number of rounds.

\begin{lemma}\label{lem:constant-rounds-epoch}
Each epoch consists of $O(1)$ communication rounds.
\end{lemma}

\begin{proof}
By Lemma~\ref{lem:parameter-setting} and Corollary~\ref{cor:slow-down}, the duration of the epoch, $X$, is in terms of $\nfast$, $\nslow$,
$\adv$ and $\Delta$. Since a communication round lasts at most $\Delta$ time and $\nfast$, $\nslow$, and $\adv$ are constants,
an epoch consists of $O(1)$ communication rounds.
\end{proof}

We now show that, with high probability, the conditions stated in Lemma~\ref{lem:parameter-setting} and Corollary~\ref{cor:slow-down}
can be satisfied. To do this, we first show that with high probability, after $O(\log n)$ epochs,
there will exist an epoch which has exactly
one so-far honest leader.

\begin{lemma}\label{lem:so-far-honest-leader}
After $O\left(\log{n}\right)$ epochs, there will be at least one epoch in which there exists exactly one leader and that leader is honest.
\end{lemma}

\begin{proof}
At the beginning of epoch $\ell$, at most $f$ replicas are Byzantine when the leader is chosen. Therefore, the probability that
an already-Byzantine node is chosen is $\frac{f}{2n} < \frac{n/3}{2n} = \frac{1}{6}$. Thus, the probability that a Byzantine node
is chosen to be a leader for every epoch after $c \log{n}$ epochs is $\left(\frac{1}{6}\right)^{c\log{n}} = \Theta\left(\frac{1}{n^c}\right)$.
Thus, with high probability, after $c\log{n}$ epochs, there will exist at least one epoch where no Byzantine replicas are
elected as leaders. By the Chernoff bound, the probability that more than one leader is elected in every epoch after $c \log{n}$ epochs is
$\leq \exp(-1/6)^{c\log{n}} = \Theta\left(\frac{1}{n^c}\right)$. The probability that no leaders are elected after $c \log{n}$ rounds is
$< \left(1 - \exp(-1/6)\right)^{c\log{n}} = \Theta\left(\frac{1}{n^c}\right)$. By the union bound, the probability that any of the above three
bad cases occur after $c\log{n}$ rounds is bounded by $\Theta\left(\frac{1}{n^c}\right)$ for all $c > 0$. Thus, with high probability,
there exists at least one round in which there exists exactly one leader and that leader is honest.
\end{proof}

\begin{lemma}\label{lem:majority-honest}
Suppose that the number of Byzantine replicas, $f$ is given by $f < \eps_f n$ for some constant $\eps_f < \eps$ provided $\eps < 1/3$ (in Fig.~\ref{fig:proposal-consensus}). Then, there exist an arbitrarily small constant $0 < \eps' < 1$ such that
after $O\left(1\right)$ epochs, there will be at least one round where all committees have $\geq \frac{2\log^2{n}}{3}$ honest replicas in each committee,
and there are $< \frac{\log^2{n}}{3}$ Byzantine replicas
in each committee with probability $1 -  O\left(\exp\left(-\eps'^2 (1-\eps_f)\log^2 n/9(1-\eps)\right)\right)$
for some constants $0 < \eps' < 1$, $0 < \eps < 1/3$, and $0 < \eps_f < \eps$.
\end{lemma}

\begin{proof}
The expected number of honest replicas that will be chosen for any committee is given by
$(n - f)\left(\frac{2\log^2 n}{3(1-\eps)n}\right) = \frac{2(1-\eps_f)\log^2 n}{3(1-\eps)}> \frac{2\log^2 n}{3}$ since $\eps_f < \eps$. By the Chernoff bound, the
probability that \emph{less than} $(1-\eps_1)\left(\frac{2(1-\eps_f)\log^2 n}{3(1-\eps)}\right)$ honest replicas are chosen into the committee is $< \exp\left(-\eps_1^2 (1-\eps_f)\log^2 n/3(1-\eps)\right)$.
In order for the number of honest replicas to be $\geq \frac{2\log^2 {n}}{3}$, we must have $(1-\eps_1)\left(\frac{2(1-\eps_f)\log^2 n}{3(1-\eps)}\right) \geq \frac{2\log^2{n}}{3}$. Thus,
we obtain $\eps_1 \leq 1-\frac{1-\eps}{1-\eps_f}$. Since $\eps_f < \eps$, there always exist values of $\eps < 1/3$ and $0 < \eps_1 < 1$ such that the condition is satisfied.
The probability that after $O\left(1\right)$ epochs there exists an epoch with
$\geq \frac{2\log^2{n}}{3}$ honest replicas in each committee is then given by $1 - O\left(\exp\left(-\eps_1^2 (1-\eps_f)\log^2 n/3(1-\eps)\right)\right)$.

The expected number of Byzantine replicas that will be chosen for any committee is given by
$f\left(\frac{2\log^2 n}{3(1-\eps)n}\right) \leq \eps_f
n\left(\frac{2\log^2 n}{3(1-\eps)n}\right) = \frac{2\eps_f\log^2 n}{3(1-\eps)}$.
By the Chernoff bound, the probability that $> (1 + \eps_2)
\left(\frac{2\eps_f\log^2 n}{3(1-\eps)}\right)$ replicas in the committee
are Byzantine replicas is given by $< \exp(-2\eps_2^2 \eps_f \log^2 n/9(1-\eps))$.
In order for the number of Byzantine replicas to be $< \frac{2\log^2 n}{3}$,
we must have $ (1 + \eps_2) \left(\frac{2\eps_f\log^2 n}{3(1-\eps)}\right)<
\frac{\log^2 n}{3}$. Solving, we obtain $\eps_2 < \frac{1-\eps}{2\eps_f} - 1$.
Since $\eps_f < \eps$ and $\eps < 1/3$, there always exists a value
$0 < \eps_2 < 1$ that satisfies this inequality.
The probability that after $O\left(1\right)$ epochs there exists an epoch where
$< \frac{\log^2{n}}{3}$ Byzantine replicas are in each committee is
then given by $1 - O\left(\exp(-2\eps_2^2 \eps_f \log^2 n/9(1-\eps))\right)$.

The probability that both conditions are satisfied is
$$1 -  O\left(\exp\left(-\eps_1^2 (1-\eps_f)\log^2 n/3(1-\eps)\right) + \exp(-2\eps_2^2 \eps_f \log^2 n/9(1-\eps))\right).$$ Thus, there exist constants $0 < \eps' < 1$, $\eps_f < \eps$, and $\eps < 1/3$
where the probability that both conditions are satisfied is $1 -  O\left(\exp\left(-\eps'^2 (1-\eps_f)\log^2 n/9(1-\eps)\right)\right)$.
\end{proof}

\begin{corollary}\label{cor:majority-honest}
    Suppose that the number of Byzantine replicas, $f$, is given by $f < \eps_f n$
    for some constant $\eps_f < \eps$ provided $\eps < 1/3$ (in Fig.~\ref{fig:proposal-consensus}).
    With high probability, after $O(1)$ epochs, there will be at least one epoch
    where all committees have $\geq \frac{2\log^2{n}}{3}$ honest replicas in each committee,
    and there are $< \frac{\log^2{n}}{3}$ Byzantine replicas
    in each committee.
\end{corollary}

\begin{proof}
By Lemma~\ref{lem:majority-honest}, the probability that the conditions of this corollary are satisfied given constants $\eps$, $\eps_f$, and $\eps'$ is $1 - O\left(\exp\left(-\eps'^2 (1-\eps_f)\log^2 n/9(1-\eps)\right)\right)$. Since $\eps$, $\eps'$, and $\eps_f$ are constants, the probability that the conditions of this corollary are satisfied is $1 - o\left(\frac{1}{n^c}\right)$ for any constant $c \geq 1$.
\end{proof}

\begin{lemma}\label{lem:bad-dominate}
After $O(\log^c n)$ epochs, for any constant $c \geq 1$,
the probability that the result of the selection of replicas
for committees gives $2(1+\eps_m) \eps_f \left(\frac{2\log^2n}{3(1-\eps)}\right)
+ (1+\eps_m)(1-\eps_f)\left(\frac{2\log^2 n}{3(1-\eps)}\right)
< \frac{4\log^2 n}{3}$ votes for all committees of an epoch is
$1 - O\left(\log^c n \exp\left(-\frac{2\eps_m^2 (1-\eps_f)
\log^2 n}{9(1-\eps)}\right)\right)$ for constants $0 < \eps_m < 1, 0< \eps_f< \eps, 0<\eps<1/3$.
\end{lemma}

\begin{proof}
Suppose there exists at least two proposals made by leaders. By the Chernoff bound,
the probability that Byzantine replicas cast $> 2(1+\eps_m) \eps_f
\left(\frac{2\log^2n}{3(1-\eps)}\right)$ votes
is bounded by $\exp(-4\eps_m^2 \eps_f \log^2 n/9(1-\eps))$, and the probability
that honest replicas cast $> (1+\eps_m)(1-\eps_f)\left(\frac{2\log^2 n}{3(1-\eps)}\right)$
votes is bounded by $\exp(-2\eps_m^2 (1-\eps_f) \log^2 n/9(1-\eps))$. By the union bound,
the probability that both types of votes are upper bounded is then
$1 - O\left(\log^c n \exp\left(-\frac{2\eps_m^2 (1-\eps_f) \log^2 n}{9(1-\eps)}\right)\right)$.
What remains to be shown is that there exist constants such that
$2(1+\eps_m) \eps_f \left(\frac{2\log^2n}{3(1-\eps)}\right) +
(1+\eps_m)(1-\eps_f)\left(\frac{2\log^2 n}{3(1-\eps)}\right) < \frac{4\log^2 n}{3}$.
Since $\eps_f < \eps$ and $\eps < 1/3$, we solve for $\eps_m$ from the expressions
to obtain $\eps_m < \frac{2(1-\eps)}{1+\eps_f} - 1$. Since $\eps_f < \eps < 1/3$,
$\frac{2(1-\eps)}{1+\eps_f} > 1$.
Thus, there exist constants $0 < \eps_m < 1, \eps_f < \eps, \eps < 1/3$ such that the inequality is satisfied.
\end{proof}

\begin{corollary}\label{cor:bad-dominate}
With high probability, after $O(\log^c n)$ epochs for any $c > 0$,
no epoch has all three committees have $\geq \frac{4\log^2 n}{3}$ votes.
\end{corollary}

\begin{proof}
Lemma~\ref{lem:bad-dominate} shows that $\geq \frac{4\log^2 n}{3}$
votes occur with $O\left(\log^c n \exp\left(-\frac{2\eps_m^2 (1-\eps_f) \log^2 n}{9(1-\eps)}\right)\right)$ probability
for constants $\eps_m, \eps_f, \eps$ which means that
with probability $1-o\left(\frac{1}{n^c}\right)$ for all constants $c$, this does not occur.
\end{proof}

Using our lemmas above, we now prove the safety and liveness of our consensus protocol.

\begin{lemma}\label{lem:proposal-safety}
Our protocol maintains safety with high probability.
\end{lemma}

\begin{proof}
By Corollary~\ref{cor:bad-dominate}, with high probability, after $O(\log^c n)$ rounds, no round has all three committees have $\geq \frac{4\log^2 n}{3}$ votes. We prove that this means that no two different proposals will be committed by different honest replicas. During the round in which there are $< \frac{\log^2 n}{3}$ Byzantine replicas and $\leq \frac{2\log^2n}{3}$ honest replicas, it is impossible to reach the threshold of $\geq \frac{2\log^2{n}}{3}$ votes on \emph{two} different proposals even if all Byzantine replicas double vote. Furthermore, no honest replica will ever vote for two different proposals in the same round. Without the necessary votes, no two different proposals will be committed by two different honest replicas. Thus, safety is maintained with high probability.
\end{proof}

\begin{lemma}\label{lem:liveness}
Our protocol reaches consensus in $O(\log n)$ epochs with high probability.
\end{lemma}

\begin{proof}
By Lemma~\ref{lem:so-far-honest-leader} and Corollary~\ref{cor:majority-honest}, there will be at least one round after $O(\log n)$ rounds where the conditions of Corollary~\ref{cor:slow-down} will be satisfied.
Therefore, after $O(\log n)$ rounds, with high probability, our protocol reaches consensus.
\end{proof}

\begin{proof}[Proof of Theorem~\ref{thm:general-final}]
The safety and liveness of our protocol are proven by Lemmas~\ref{lem:proposal-safety} and~\ref{lem:liveness}.
Furthermore by Lemma~\ref{lem:constant-rounds-epoch}, each epoch consists of $O(1)$ rounds; thus,
our protocol terminates in $O(\log n)$ rounds with high probability.
Since committees have size $O(\log^2 n)$, our protocol uses
at most $O\left(\log^3 n\right)$ multicasts.
\end{proof}

\section{Sublinear Clock Synchronization with Adaptive Adversaries and Randomly Dropped Messages}\label{sec:sublinear-bba}

Given the previous impossibility result in~\cref{sec:impossible-bba}, we cannot hope
to achieve BBA with sublinear multicasts with high probability
in the partially sychronous (GST) network model as defined in~\cite{dls88}
with adaptive adversaries.  Instead, we use a slightly
different network model. We define this model as the \emph{partially
synchronous with randomly dropped messages} model. This model may have
practical applications as a model that represents an unreliable/faulty network.

\begin{definition}[Partially Synchronous with Randomly Dropped Messages Model]\label{def:psrdm}
For the asynchronous phase of the \emph{partially synchronous with randomly dropped messages model}, the adversary can only choose to
perform the following actions on the network:
\begin{enumerate}
\item drop messages with some constant probability $0 \leq p < 1$ (i.e.\ each individual message has a probability $p$ of being dropped),
\item delay messages by delay at most $\Delta$, and
\item delay processors by delay at most $\Phi$.
\end{enumerate}

All other characteristics of the model follow that of the partially synchronous (GST) model.
\end{definition}

In this network model, we prove the following theorem:

\begin{theorem}\label{thm:sublinear-bba}
    There exists a communication-efficient
    BBA protocol in the partially synchronous with randomly dropped messages model
    that reaches agreement with high probability with respect to $n$ and $\secp$.
\end{theorem}

Our protocol is based on an adaptation, provided in~\cref{sec:bba},
of the linear-multicast protocol given in~\cite{cc19}
for the partially synchronous model. Specifically, we formulate a novel clock synchronization
procedure for both the partially synchronous model (using linear multicasts) and the partially
synchronous with randomly dropped messages model (using sublinear multicasts). This clock
synchronization procedure uses the three-step commit rule of HotStuff~\cite{hotstuff} and
may be of independent interest for use in other protocols. The details on our linear multicast
protocol in the partially synchronous (GST) model can
be found in~\cref{sec:bba} and details on our sublinear multicast protocol in
the partially synchronous with randomly dropped messages model can be found in~\cref{sec:sublinear-bba}.

We first assume in our protocols that knowledge of $p$ is given at the time of formulation of the protocol (i.e.\ $p$ can be assumed to be a known constant in our protocols).

Keeping such challenges described in \cref{sec:bba} in mind, we first provide our revised round synchronization protocol below.

\begin{figure}
\begin{mdframed}[hidealllines=false,backgroundcolor=gray!30]
\paragraph{Protocol for replica $i$:}
\begin{enumerate}
  \item Replica $i$, set $C_i = 0$.
  \item While protocol not terminated, for an honest replica $i$:
  \begin{enumerate}
      \item \underline{Decide whether part of the \emph{round proposal committee}, i.e.\ check if $\cryptosort(\sk{i}, C_i) < D_0$.}
  	\item Multicast $C_i + 1$, certificate for $C_i$ (when $C_i=0$, no certificate is needed), \underline{and proof of committee membership each repeatedly $1/p$ times}.
  	\item \underline{Wait $\Delta + \Phi$ time or after receiving $>3\lambda$ round proposals}. Record all valid rounds seen (i.e.\ rounds $R$ with a valid certificate for $R-1$). If $i$ sees a certificate for a round $j > C_i$, $i$ updates its $C_i \leftarrow j$. \underline{If the number of rounds seen is $> 3\lambda$}, determine the smallest round $R$ greater than $C_i$, is smaller than $T = O\left(3^{\log^c n}\right)$, and with a valid certificate for $R-1$; let this round be $S$.
    \item \underline{If $\cryptosort(\sk{i}, \text{`tentative'}, S) < D_0$,
        then $i$ is a member of the} \underline{\emph{tentative round voting comittee}}:
        then multicast a vote for $S$ \underline{repeatedly $\frac{1}{1-p}$ times}.
  	\item Count votes. If any round receives \underline{$2\lambda+ 1$} votes for $S$,
        set \emph{tentative round} to this round.
  	Set $T_i$ to be the new tentative round.
    \item \underline{If $\cryptosort(\sk{i}, \text{`pre-confirmed'}, T_i) < D_0$,
        then $i$ is a member of the} \underline{\emph{pre-confirmed round voting comittee}}:
        then multicast a vote for $T_i$ \underline{repeatedly $\frac{1}{1-p}$ times}.
  	\item Count votes. If $T_i$ receives \underline{$2\lambda+ 1$} votes, set \emph{pre-confirmed round} to this round.
  	Set $C_i$ to be the new pre-confirmed round.
\item \underline{If $\cryptosort(\sk{i},\text{`confirmed'}, C_i) < D_0$,
    then $i$ member of \emph{confirmed round voting committee}}:
    then multicast vote for pre-confirmed round $C_i$ \underline{repeatedly $1/p$ times}.
  	\item Upon receiving \underline{$2\lambda + 1$} votes for pre-confirmed round $C_i$,
        set \emph{confirmed round} to this round. Perform the rest of the protocol only if confirmed round is set.
  	\item Timeout and restart with $C_i + 1$ again if any of the above steps take longer than $2\Delta + 2\Phi$ time
        or if any of $S$, $T_i$, or $C_i$ does not receive enough votes.
  \end{enumerate}
\end{enumerate}
\end{mdframed}
\caption{Sublinear Multicast Clock Synchronization in
Partially Synchronous with Randomly Dropped Messages Model.
Underlined portions of the protocol are instructions unique to
the partially synchronous with randomly dropped messages model.}\label{alg:sublinear-clock-alg}
\end{figure}

\subsection{Sublinear Clock Synchronization Protocol Analysis}

We now analyze our sublinear round protocol (which is also an honest total multicast protocol) that synchronizes to a round
that is not advantageous for the adversary, keeping in mind the challenges described in \cref{sec:bba}.

We first show that an adversarially dominated committee is unlikely.
Suppose that the asynchronous part of the protocol lasts at most $T = O\left(3^{\log^c n}\right)$ rounds, then using this, we
can prove the following lemma:

\begin{lemma}\label{lem:committee-dominates}
Suppose $D_0 = \frac{2\log^d n}{3(1-\eps)n}$ (where $d > c$) for some $\eps < 1/3$ and the number of Byzantine replicas, $f$,
in the network is $f <\eps' n$ where $\eps' < \eps$, the probability that a
round where the adversary controls at least $\frac{(1+\eps'')\log^c n}{3}$
for $\eps'' < 1$ replicas in the
committee exists within $T = O\left(3^{\log^c n}\right)$ rounds is $O\left(\frac{1}{n^{d-1}}\right)$.
\end{lemma}

\begin{proof}
The probability that any one replica joins the committee is $\frac{2\log^d n}{3(1-\eps)n}$.
Given $f < \eps' n$ Byzantine replicas where $\eps' < \eps$, the expected number of Byzantine replicas in
any committee is $\frac{2\eps' \log^d n}{3(1-\eps)}$.
Thus, by the Chernoff bound,
the probability that any committee has more than $(1+\eps'')\frac{2\eps' \log^d n}{3(1-\eps)} < (1+\eps'') \frac{\log^d n}{3}$
replicas is given by $< \exp(-\eps''^2\log^d n/9)$.
Thus, the probability that at least $\frac{(1+\eps'')\log^c n}{3}$
Byzantine replicas are in any committee
after $T$ rounds is upper bounded by
$O\left(Te^{-\frac{\eps''^2 2\eps'\log^{d} n}{9(1-\eps)}}\right) =
O\left(\frac{3^{\log^c n}}{e^{-\frac{\eps''^2 2\eps'\log^{d} n}{9(1-\eps)}}}\right) =
o\left(\frac{1}{n^{a}}\right)$ for all constant $a$.
\end{proof}

Now we bound the probability that during the asynchronous period less than $(1-\eps'')\frac{2\log^d n}{3}$
honest replicas receive the
lowest valid round that is multicasted.

\begin{lemma}\label{lem:min-proposed-round}
Suppose $D_0 = \frac{2\log^d n}{3(1-\eps)n}$ for some $\eps < 1/3$ and the number of Byzantine replicas, $f$,
in the network is $f <\eps' n$ where $\eps' < \eps$,
the probability that less than $(1-\eps'')\frac{2\log^d n}{3}$ honest replicas
in a committee receive
the minimum round that is multicasted is $O\left(\frac{1}{n^c}\right)$
for any $c > 1$.
\end{lemma}

\begin{proof}
The probability that any message is dropped is $p$.
But each honest replica multicasts each message $1/p$ times.
Thus, the expected number of honest replicas in a committee that
receive the minimum round that is multicasted is $\frac{2(1-\eps')\log^{d}n}{3(1-\eps)}$.
Since $\eps' < \eps < 1/3$, this value is lower bounded by $\frac{2\log^d n}{3}$.
Hence, by the Chernoff bound,
the probability that less than $(1-\eps'')\frac{2\log^d n}{3}$ honest replicas
do not receive the smallest valid round multicasted is
$O\left(\frac{1}{n^c}\right)$ for $c > 1$ for
any constant $\eps'' > 0$ used in the Chernoff bound.
\end{proof}

Using the above we can now prove the safety of our round synchronization protocol during the asynchronous period of our
network.

\begin{lemma}\label{lem:smaller-epoch}
Given committees of size $O(\log^{d} n)$ and $T = \Theta\left(3^{\log^c n}\right)$,
an honest replica never confirms an epoch smaller than an epoch
already confirmed by another honest replica with high probability.
\end{lemma}

\begin{proof}
In order for an honest replica to confirm an epoch smaller
than the largest confirmed epoch by any honest replica,
a sequence of events must occur. First, the replica must
not have received the votes confirming the previous larger epoch.
Second, it must not have seen the previous epoch proposed by the
proposal committee. Finally, a large enough portion of the
committee must not have seen the certificate multicasted for
epochs greater than the smaller epoch in order to confirm
the new smaller epoch (as well as the previous multicasts).
Given a set of $\log^{d} n$ committee members,
the probability that all of the above occurs is upper bounded by the probability that
enough honest replicas do not see the smallest epoch
multicasted and the adversary controls a large enough fraction of the
committees. The probability that both occur by the union bound on Lemmas~\ref{lem:committee-dominates} and~\ref{lem:min-proposed-round}
is $O\left(\frac{1}{n^c}\right)$ for any $c > 1$. Thus, with high probability when the asynchronous period lasts
for $T = O\left(3^{\log^c n}\right)$ rounds, safety is preserved.
\end{proof}

The remaining proofs for round synchronization follow closely that of the warmup protocol in \cref{sec:warmup-protocol-analysis}.

\subsection{Sublinear BBA in the Partially Synchronous with Randomly Dropped Messages Model}\label{sec:sublinear-bba-full}

Below, we use our sublinear multicast clock synchronization protocol to obtain a sublinear multicast BBA protocol in the
partially synchronous with randomly dropped messages model. Because our protocol is very similar to our linear multicast protocol
presented in \cref{sec:bba}, we only underline the portions of the protocol that differ in this case.

\begin{figure}
\begin{mdframed}[hidealllines=false,backgroundcolor=gray!30]
\paragraph{Protocol for replica $i$:}
\begin{enumerate}
\item Set $T_i = 1$ and $F_i = 1$ at the start of the protocol
before receiving \emph{any messages} from round leaders (and before sending $\ack$s).
Initialize $b_i^*$ to the bit received initially as input before the protocol starts.
\item The following is performed repeatedly until replica $i$ commits to a bit (outputs a bit):
\begin{enumerate}
\item Run the round synchronization process detailed in Figure~\ref{alg:sublinear-clock-alg}. Only proceed with the rest of the protocol after becoming synchronized to a round.
Let $S$ be this confirmed round.
\item Flip a random (fair) coin to determine a bit, $b$.
    Then, check if $\cryptosort(\sk{i}, \prop, S, b) < \pt$ for some value
$\pt$ to be determined later in our analysis.
\item If $\cryptosort(\sk{i}, \prop, S, b) < \pt$, multicast $(\prop, S, b)$ and a proof.
\item After receiving a valid propose $(\prop, S, b')$ message (and proof $\pi_{S}$):
\begin{enumerate}
\item Wait $\delta$ time (for some $\delta$ to be defined in the analysis) to see if it receives another unconfirmed proposal $(\prop, S, b'')$ where $b' \neq b''$. If $i$ receives such a proposal, then
$i$ does nothing this round.
\item Otherwise, set $b_i^* \coloneqq b'$ if $F_i = 0$
\item \underline{If $\cryptosort(\sk{i}, \texttt{Multicast}, S, b_i^*) < D_1$}, multicast $(\ack, S, b_i^*)$.
\end{enumerate}
\item If received \underline{$\geq \floor{2 \log^d n/3}$} $\ack$s $(\ack, S, b')$ (from different replicas) where $b_i^* = b'$, set $F_i \coloneqq 1$.
\item If received \underline{$\geq \floor{2\log^d n/3}$} $\ack$s $(\ack, S, b')$ (from different replicas) where $> f$ of the $\ack$s are for $b_i^* \neq b'$, set $F_i \coloneqq 0$.
\item If after $\delta$ delay, $i$ eventually received at least \underline{$\floor{2\log^d n/3}$} $\ack$s $(\ack, S, b)$, then increment $T_i \xleftarrow T_i + 1$.
\item At the end of $K$ rounds, i.e.\ when $T_i = K$ (for some $K$ to be determined later in our analysis), if $F_i = 1$, commit to $b_i^*$ and return $b_i^*$ as
output bit.
\item After completing the above protocol, start the round synchronization process again.
\item If after waiting $5\delta$ time, no further step can be taken during any point of the above protocol, start with the next iteration at the beginning of this loop.
\end{enumerate}
\end{enumerate}
\end{mdframed}
\caption{Polylogarithmic multicast simple round synchronization protocol.}\label{alg:sublinear-bba-alg}
\end{figure}

\subsection{Analysis}

We first prove the following lemma regarding our protocol which will help us obtain
our final result on the communication complexity and number of rounds (after GST) our protocol requires.

\begin{lemma}\label{lem:enough-honest}
    Let $D_1 = \frac{2\log^d n}{3(1-\eps)n}$, after $O(\poly \log n)$ rounds after GST, the probability that
    all epochs have at least one committee which has $< \floor{\frac{2\log^d n}{3}}$ honest replicas is
    $O\left(\frac{1}{n^c}\right)$ for all $c > 1$.
\end{lemma}

\begin{proof}
    The expected number of honest replicas in each committee is at least $\frac{2(1-\eps')\log^d n}{3(1-\eps)n}$
    since we expect the number of honest replicas to be $\geq (1 - \eps')n$. Since, by definition,
    $\eps' < \eps$, the expected number of honest replicas in each committee is then at least $\frac{2c\log^d n}{3}$ where
    $c = \frac{1-\eps'}{1-\eps}$. Then, by the Chernoff bound, the probability that
    less than $(1-\eps_c)\left(\frac{2c\log^d n}{3}\right)$ honest replicas are in any committee is
    $< \exp\left(-\frac{2\eps_c^2c \log^d n}{3}\right)$. When $(1-\eps_c)\left(\frac{1-\eps'}{1-\eps}\right) \geq 1$,
    $(1-\eps_c)\left(\frac{2c\log^d n}{3}\right) \geq \frac{2\log^d n}{3}$. There exists such a constant
    $0 < \eps_c < 1$ where this constraint is satisfied since $\eps' < \eps$ (i.e.\ when $\eps_c \leq 1- \frac{1-\eps}{1-\eps'}$)
    The probability that any of the three committees in an epoch have less than $\frac{2\log^d n}{3}$
    honest replicas in the committee is, by the union bound, $3\exp\left(-\frac{2\eps_c^2c \log^d n}{3}\right)$.
    Hence, after $O(\poly \log n)$ rounds, the probability that all epochs have at least one committee which has
    $< \floor{\frac{2\log^d n}{3}}$ honest replicas is upper bounded by $O\left(\frac{1}{n^c}\right)$ for all $c > 1$.
\end{proof}

Next, we prove that with high probability Byzantine replicas cannot vote for both $0$ and $1$ in all
three committees for any epoch after $O(\log^{d} n)$ epochs.

\begin{lemma}\label{lem:both-bits}
    After GST, given $D_1 = \frac{2\log^d n}{3(1-\eps)n}$,
    after $O(\log^{d} n)$ epochs, with high probability, Byzantine replicas cannot vote for both $0$ and $1$ in all
    three committees (such that the votes reach the threshold for both bits) for any of the $O(\log^d n)$ epochs.
\end{lemma}

\begin{proof}
    The expected number of votes mined by Byzantine replicas is upper bounded by the expected number of votes
    mined by all replicas, which is at most $\frac{2 \log^d n}{3(1-\eps)}$.
    In order to vote for both $0$ and $1$ and successfully reach the threshold, Byzantine replicas
    and honest replicas together need to mine
    at least $\frac{4\log^d n}{3}$ votes. Since we defined $\eps < 1/3$, the expected
    number of votes is $< \frac{4\log^d n}{3}$ (where we assume that Byzantine replicas attempt to
    vote twice). Then, by the Chernoff bound, the probability that
    any committee has $\geq \frac{4\log^d n}{3}$ votes is at least
    $\exp\left(-\frac{\left(\frac{4}{3(1+\eps')}-1\right)^2\log^d n}{3}\right)$.
    Thus, with high probability after $O(\log^d n)$ rounds, no epoch has all three committees
    have enough votes such that Byzantine replicas can vote for both $0$ and $1$ and reach the threshold.
\end{proof}

Now, we can compute the probability that a single, honest leader is elected after $O(\log n)$ rounds.

\begin{lemma}\label{lem:leader-honest}
    After GST, given $D_0 = \frac{1}{12n}$ and where $\eps < 1/3$,
    after $O(\log n)$ epochs, with high probability, a single honest leader is elected.
\end{lemma}

\begin{proof}
    With probability $\leq 1 - \frac{1-\eps'}{12}$, no honest leaders are
    elected in round $i$. The probability that more than a single leader is elected in round $i$ is
    given (by the Chernoff bound) to be $< \exp(-121/36)$. Thus, the probability that round $i$
    either has no honest leaders or more than one leader is, by the union bound,
    $1 - \frac{1-\eps'}{12} + \exp(-121/36) < 1$.
    Thus, after $O(\log n)$ rounds, a single honest leader is elected with high probability.
\end{proof}

Now we prove that, with high probability, honest nodes cannot commit to two different
bits during the asynchronous period. Furthermore, we prove that with high probability,
if all honest nodes start with the same bit, no honest node will commit to the other bit.

\begin{lemma}\label{lem:same-bit}
    No two honest replicas will commit two different bits during the asynchronous period.
\end{lemma}

\begin{proof}
    In order for a honest replica to commit to a bit, it must receive a proposal and committee votes from
    three committees. The expected number of honest replicas that receive the proposal and all
    (honest) committee votes is $(1-\eps')n$. Thus, by the Chernoff bound, the probability
    that $< (1-\eps_c)(1-\eps')n$ honest replicas receive all proposals and vote messages
    is $< \exp\left(-\frac{\eps_c^2(1-\eps')n}{2}\right)\leq \exp\left(-\frac{\eps_c^2 n}{3}\right)$.
    The expected number of votes that come from the honest replicas that do not receive all messages
    and Byzantine replicas (assuming all Byzantine replicas attempt to vote twice) is
    $< \frac{4\eps'\log^d n}{3(1-\eps)}+\frac{2\eps_c \log^d n}{3(1-\eps)}$.
    Let $S$ be the expected number of such votes. By setting
    $\eps_c < 1 - \eps - 2\eps'$, the expected number of such votes is $S < \frac{2\log^d n}{3}$.
    By the Chernoff bound, the probability that more than $(1 + \eps_1)S$ votes are from
    Byzantine replicas or honest replicas that do not receive all messages is
    $< \exp\left(-\frac{2\eps_1^2 \log^d n}{3}\right)$.
    Hence, assuming $T = \Theta\left(3^{\log^c n}\right)$,
    with high probability, no honest replica will commit the other bit during the asynchronous period.
\end{proof}

\begin{corollary}\label{cor:two-bits}
    If all honest replicas start with the same bit as their input, then, with high probability,
    no honest replica will commit to the other bit during the asynchronous (with randomly
    dropped messages) period.
\end{corollary}

Using the above lemmas, we prove our final result.

\begin{theorem}\label{thm:random-final}
    Let $D_0 = \frac{1}{12n}$, $D_1 = \frac{2\log^d n}{3(1-\eps)n}$ and
    assuming $T = O\left(3^{\log^c n}\right)$, after GST, after $O(\poly \log n)$ rounds,
    agreement is reached with high probability using $O(\log^d n)$ multicasts where $d > c$.
\end{theorem}

\begin{proof}
    By~\cref{lem:same-bit} and~\cref{cor:two-bits}, no two honest replicas will commit to different bits during
    the asynchronous period with high probability. Thus, we are concerned with only
    reaching agreement after GST. By~\cref{lem:committee-dominates}, with high probability,
    no committee after GST is dominated by Byzantine replicas. Furthermore, \cref{lem:smaller-epoch}
    shows that no replica confirms an epoch smaller than an epoch already confirmed, hence,
    progress is made after GST. By~\cref{lem:enough-honest}, with high probability after GST, there
    exists an epoch after $O\left(\poly \log n\right)$ rounds where all committees have
    sufficiently many honest replicas. Furthermore, by~\cref{lem:both-bits}, after GST, there
    does not exist an epoch where Byzantine replicas can successfully vote for
    both bits before agreement is reached with high probability. Finally,
    by~\cref{lem:leader-honest}, within $O\left(\poly \log n\right)$
    rounds after GST, a single honest leader is elected with high probability.
    Thus, within $O\left(\poly \log n\right)$ rounds, $O(\log n)$ honest leaders
    are elected with high probability. The remaining parts of the proof directly follow
    that provided in~\cref{sec:warmup-protocol}.
    Hence, we prove our lemma statement.
\end{proof}

\section*{Acknowledgements}

We thank Madars Virza for very helpful discussions and comments.

\bibliographystyle{alpha}
\bibliography{ref}
\appendix
\section{Deferred Proofs from~\cref{sec:impossible-bba}}\label{app:impossible-bba}

\begin{proof}[Proof of~\cref{lem:lower-partially-sync-bba}]
By definition of a partially synchronous network and an adaptive adversary, the adversary has the power to selectively delete
messages during the asynchronous phase. Thus, we simulate the strategy presented in Theorem 4 of~\cite{cc19} during the
asynchronous phase to prove our lower bound.
Because the adversary can selectively choose to change $\Delta$ to any value for any honest replica,
the adversary can at any point during the asynchronous
phase choose to set $\Delta$ to be the synchronous value; thus,
honest replicas have no way to distinguish between when they are in the asynchronous phase or in the synchronous phase.
The adversary then selectively chooses to send messages or set $\Delta$ to be infinite
for messages they want to drop, following the strategy given in Theorem 4 of~\cite{cc19}.
Thus, BBA cannot be reached in $o(n)$ multicasts in expectation with any partially synchronous protocol.

We now show that this statement holds with high probability. As we show in~\cref{sec:bba}, there exist a $O(n)$ multicast protocol
that solves BBA with high probability in the partially synchronous (GST) model.
Thus, suppose for the sake of contradiction that there exists a $o(n)$ multicast protocol
that solves BBA with high probability. Then, in expectation $o(n)$ multicasts are necessary to solve BBA also since the protocol
can always just use the $O(n)$ protocol if the number of multicasts start becoming too large.
Thus, we have a contradiction with our proof above which uses Theorem 4
of~\cite{cc19} and such a high probability protocol does not exist.
\end{proof}

\section{Partially Synchronous Binary Byzantine Agreement with Adaptive Adversaries}
\label{sec:bba}

Although in~\cref{sec:impossible-bba}, we show that
sublinear-multicast binary Byzantine agreement (BBA) is impossible in
the partially synchronous model as defined by~\cite{dls88}.
A natural question to ask is whether we can achieve a linear multicast
BBA protocol under adaptive adversaries in the proper partially synchronous model
and then modify it
so that it achieves sublinear multicasts for an alternative
partially synchronous network model.
Our protocol is inspired by the~\cite{cc19} protocol
in the synchronous model for BBA.
We first state the simple version with $O(n)$ multicast complexity.
In \cref{sec:sublinear-bba}, we modify this protocol
to have sublinear multicast complexity in a slightly different network model.
As in~\cite{cc19} and~\cref{sec:general-consensus}, we assume
for simplicity that oracle $\cryptosort$ exists.
Possible instantiations of $\cryptosort$ are given in~\cref{app:vrfs}.

Each replica maintains the following states: $T_i$, $b_i^*$, $F_i$, $C_i$, and $A_i$. $A_i$ is a counter for maintaining the number of leader messages that a replica
has $\ack$ed. $b_i^*$ is the current bit that the replica is maintaining. $F_i$ is a state counter that maintains whether the bit
will be changed to the newest bit received from the current leader or whether the previous bit will be copied over.
$T_i$ is the epoch number of the next round the replica will propose.
$C_i$ is the epoch number of the last confirmed round that the replica has seen.
A \emph{confirmed round proposal} is a round where replica $i$ has seen more than $2n/3$ $\ack$s during that round.  This set of $> 2n/3$ $\ack$s comprises a \emph{certificate} for the round.

\paragraph{Challenge 1: Preventing rewinding and fast forwarding} It is crucial that the rounds progress so that the adversary cannot influence leader selection. Since, in the partially synchronous model, we have no way to synchronize rounds during the
asynchronous periods of time, in order to prevent the adversary from advancing the rounds much too quickly
(affording them an advantage in terms of being selected leader), we must be able to synchronize to the highest round
that any honest replica is currently in. We must avoid resetting the round to a previous round before a round
in which an honest replica has already sent out its proposal since the sequence of leaders will be known to
the adversary for that period of time. To fix this problem, we need to be able to synchronize the replicas to the same round
\emph{before we send out additional messages}.

An adversary might also try to reset the
epoch number to a round where there is no honest leader. In this case,
no bits will be proposed and so we are potentially stuck in a repeating cycle of not being able to move forward
because the adversary would always reset to that round.
We note that this is not a problem for our protocol, since each round will be confirmed
during the epoch synchronization phase. Thus, any round where any replica decides to
query for a leader token after GST will be used at most twice before the round will \emph{never} be up for confirmation ever again.
We do this by syncing to the \emph{smallest round} after our last confirmed proposed round.
By using this procedure, we simultaneously protect against all three attacks by:
introducing fresh rounds so that the adversary does not have a schedule,
increasing the epoch number to avoid becoming stuck in a round with no honest leader,
and ensuring that the adversary
does not propose an absurdly large epoch number for their benefit.
More details can be found in \cref{sec:warmup-protocol} and \cref{sec:warmup-protocol-analysis}.

\paragraph{Challenge 2: Synchronizing replicas to the same round with processor speeds that differ by $\Phi$}
The difficulty of synchronizing the different replicas in the protocol under the adaptive adversary
model with sublinear multicast complexity is that some
`ticks' of the clock in the traditional synchronization literature~\cite{pbft,dls88} might never reach their recipients.
Furthermore, to the best of our knowledge, all known clock synchronization literature~\cite{pbft,dls88} obtains clock
synchronization in $\tilde{O}(n)$ multicast complexity.
Thus, for our protocol, we must present a novel clock synchronization technique that could be of independent interest to future
researchers. We first present a simplified version of
our linear multicast clock synchronization protocol in~\cref{alg:linear-clock-alg} where
we give a brief overview, but the details of such an implementation are provided in~\cref{app:clock-sync}.
In \cref{sec:sublinear-bba}, we modify this protocol to obtain sublinear multicast complexity in a slightly different network model
based on the partially synchronous network model.
For our more complicated protocol described in \cref{sec:sublinear-bba}, we show that with high probability our
protocol uses sublinear multicast complexity to synchronize to the same round.

\subsection{Clock Synchronization under Adaptive Adversaries}\label{app:clock-sync}

\begin{figure}[t]
\begin{mdframed}[hidealllines=false,backgroundcolor=gray!30]
  \paragraph{Replica $i$, set $C_i = 0$.}
  \paragraph{While protocol not terminated, for an honest replica $i$:}
  \begin{enumerate}
  	\item Multicast $C_i + 1$ and certificate for $C_i$ (when $C_i=0$, no certificate is needed).
    \item Wait $2(\Delta + \Phi)$ time. Record all $C_k$ seen for $k \in [2n]$.
        If $i$ sees a certificate for a round $j > C_i$, $i$ updates its $C_i \leftarrow j$.
        Determine the smallest round greater than $C_i$; let this round be $S$.
  	\item Multicast a vote for $S$.
  	\item Count votes. If any round receives $> \floor{2n/3}$ votes (where the round $\leq S$),
        set \emph{initial round} to this round.
  	\item Multicast vote for this initial round.
  	\item Count votes. If the stored initial round gets $> \floor{2n/3}$ votes (where the round $\leq S$),
        set \emph{tentative round} to this round.
  	\item Multicast vote for this tentative round.
  	\item Count votes. If the stored tentative round gets $> \floor{2n/3}$ votes (where the round $\leq S$), set \emph{pre-confirmed epoch} to this round. Set $C_i$ to be the new pre-confirmed epoch.
  	\item Multicast vote for this pre-confirmed epoch.
  	\item Upon receiving $>\floor{2n/3}$ (including own vote) for the stored pre-confirmed epoch,
        set \emph{confirmed round} to this round. Perform the rest of the protocol only if confirmed round is set.
    \item Timeout and restart with $C_i + 1$ again if any of the above steps take longer than $2(\Delta + \Phi)$ time.
  \end{enumerate}
\end{mdframed}
\caption{Shortened version of the linear multicast simple epoch synchronization protocol.}\label{alg:linear-clock-alg}
\end{figure}

We describe and prove the properties of our clock synchronization mechanism in this section. Our clock synchronization
protocol maintains a synchronized \emph{epoch number} (one can think of the epoch number to be the clock time) among all replicas
while the network remains synchronous. This protocol is inspired by the clever 3-step confirmation protocol proposed by~\cite{hotstuff}.

\begin{enumerate}
\item $i$ multicasts $\tick(S_i)$ to all replicas as a message announcing the round it is currently on.
\item $i$ waits for $2\Delta + 2\Phi$ time and keeps track of all round messages sent to it
while recording the round message with the smallest epoch number it receives containing a round that is
greater than its largest stored \emph{pre-confirmed epoch}.
\item If $i$ receives a certificate showing a greater pre-confirmed epoch
    than the largest pre-confirmed epoch it stored, then it stores the larger pre-confirmed epoch.
\item After $2(\Delta + \Phi)$ time or after it has received $\geq \floor{2n/3}$ proposed rounds,
    $i$ sends a signed \\ $\tick(S, \texttt{`tentative-vote'})$
    message to all other replicas where $S$ is the smallest
epoch number it has received a proposal for (or $S_i$) that is
greater than its largest stored pre-confirmed epoch.
\item $i$ waits and records all $\tick(S, \texttt{`tentative-vote'})$
    messages it receives. If $i$ receives at least
$2f + 1$ $\tick(S, \texttt{`tentative-vote'})$ messages for $S$,
$i$ sends a $\tick(S, \texttt{`pre-confirmed-vote'})$ message to all other replicas.
\item $i$ sets its \emph{tentative final round} to $S$ if it sends out a $\tick(S, \texttt{`pre-confirmed-vote'})$.
\item $i$ waits and records all $\tick(S, \texttt{`pre-confirmed-vote'})$ votes it receives. If $i$ receives $2f+1$ votes for $S$, then $i$ sends the
message $\tick(S, \texttt{`confirmed-vote'})$ and sets its \emph{pre-confirmed epoch} to $S$: $P \leftarrow S$.
\item If $i$ sets a pre-confirmed epoch, it resets its largest pre-confirmed epoch counter to be the new pre-confirmed epoch and records the signatures
of the replicas which voted for this new pre-confirmed epoch as the certificate for this
pre-confirmed epoch.
\item $i$ waits and records all $\tick(S, \texttt{`confirmed-vote'})$ votes it receives. If $i$ receives $2f+1$ votes for $S$, then $i$ sets its \emph{confirmed round} to $S$.
\item If $i$ sets a confirmed round, $i$ sets $S_i \leftarrow S + 1$  and proceeds with the rest of the protocol.
\item If after waiting $2(\Delta + \Phi)$ time and not receiving
    the necessary votes to set a tentative, pre-commit, or commit round, and if its
largest pre-confirmed epoch has not changed, $i$ restarts its synchronization protocol with step 1, resets $S_i = P + 1$ and multicasts $S_i$ (it resets its $S_i$ to be bigger than its most recent pre-commit round $P$ and tries again).
\end{enumerate}

We prove a set of properties for our clock synchronization protocol in the partially synchronous model under adaptive adversaries
that are too strong for our needs, but may be helpful for future work.

\begin{lemma}
Adaptive adversaries behave similarly to non-adaptive adversaries in this clock synchronization protocol.
\end{lemma}

\begin{proof}
Corrupting a node after seeing some of their messages serves no advantage since additional messages do not
impact vote counting (as honest replicas only count the votes for the smallest round larger than
their pre-committed round).
Furthemore, all nodes vote on the smallest round, thus, adaptively preventing a particular node
from voting does not affect the protocol.
\end{proof}

\begin{lemma}\label{lem:safety}
If a replica executes a part of the protocol that comes after clock synchronization,
then at least $f + 1$ honest replicas have the same or larger pre-confirmed epoch at that point in time.
\end{lemma}

\begin{proof}
Let time $t$ be the time when the next part of the protocol gets executed (meaning at least one honest replica has set a confirmed round)
with epoch number $S$.
This means that at least $2f + 1$ replicas sent a confirmed vote for the round. Then, at least $f + 1$ honest replicas
sent a confirmed vote (since at most $f$ replicas are Byzantine) and set their pre-confirmed epoch to $S$.
Suppose for the sake of contradiction that at most $f$ honest replicas have epoch number $S_i \geq S$ at time $t$.
Then, it must be the case that at least one of the $f + 1$ honest replicas voted for a round smaller than $S$
after sending a confirmed vote for $S$. This is impossible since they would have
stored $S$ as their largest pre-confirmed epoch when proposing $S_i$,
and we have reached a contradiction.
\end{proof}

We obtain as an immediate corollary:

\begin{corollary}
The protocol after clock synchronization only runs once for every round $S$.
\end{corollary}

\begin{proof}
Since the pre-commit round for $f+1$ honest replicas is at least $S$ by the time
the protocol runs for round $S$, round $S$ will never be confirmed again (and hence
the protocol will never be run again for round $S$).
\end{proof}

We now prove the liveness of the system after GST.

\begin{lemma}
After GST and assuming GST lasts for at least four voting rounds, all honest replicas will have the same confirmed round and execute the next part of the protocol.
\end{lemma}

\begin{proof}
By induction, we can show that at least one replica will propose a round that is greater than the greatest stored pre-commit round
by any honest replica. Thus, all honest replicas will vote for this proposed round. After voting three times for this round,
all honest replicas will then have the same confirmed round since all replicas will see all other replicas' votes.
\end{proof}

\begin{theorem}
All honest replicas' clocks will be synchronized to the same round after GST in four rounds, $O(\Delta)$ time, and using $O(n)$ multicasts.
\end{theorem}

Figure~\ref{alg:linear-clock-alg} is a shortened version of the clock synchronization protocol.

\subsection{Warmup Protocol}\label{sec:warmup-protocol}

Using the tools we introduced above and keeping in mind the stated challenges, we now present our complete
binary Byzantine agreement protocol that reaches agreement in $O(n)$ multicasts.
We assume here that the number of total replicas in the network is $n$ where $n \geq 3f + 1$ where $f$ is the number of Byzantine
replicas in the network.

\begin{figure}
\begin{mdframed}[hidealllines=false,backgroundcolor=gray!30]
\paragraph{Protocol for replica $i$:}
\begin{enumerate}
\item Set $T_i = 1$ and $F_i = 1$ at the start of the protocol
before receiving \emph{any messages} from round leaders (and before sending $\ack$s).
Initialize $b_i^*$ to the bit received initially as input before the protocol starts.
\item The following is performed repeatedly until replica $i$ commits to a bit (outputs a bit):
\begin{enumerate}
\item Run the epoch synchronization process detailed in Figure~\ref{alg:linear-clock-alg}.
    Only proceed with the rest of the protocol after becoming synchronized to an epoch.
    Let $S$ be this confirmed epoch number.
\item Flip a random (fair) coin to determine a bit, $b$.
    Then, check if $\cryptosort(\sk{i}, \prop, S, b) < \pt$ for some value
$\pt$ to be determined later in our analysis.
\item If $\cryptosort(\sk{i}, \prop, S, b) < \pt$, multicast $(\prop, S, b)$ and a proof.
\item After receiving a valid propose $(\prop, S, b')$ message (and proof $\pi_{S}$):
\begin{enumerate}
\item Wait $\delta$ time (for some $\delta$ to be defined in the analysis) to see if it receives another unconfirmed proposal $(\prop, S, b'')$ where $b' \neq b''$. If $i$ receives such a proposal, then
$i$ does nothing this round.
\item Otherwise, set $b_i^* \coloneqq b'$ if $F_i = 0$
\item Multicast $(\ack, S, b_i^*)$.
\end{enumerate}
\item If received $\geq \floor{2n/3}$ $\ack$s $(\ack, S, b')$ (from different replicas) where $b_i^* = b'$, set $F_i \coloneqq 1$.
\item If received $\geq \floor{2n/3}$ $\ack$s $(\ack, S, b')$ (from different replicas) where $> f$ of the $\ack$s are for $b_i^* \neq b'$, set $F_i \coloneqq 0$.
\item If after $\delta$ delay, $i$ eventually received at least $\floor{2n/3}$ $\ack$s $(\ack, S, b)$, then increment $T_i \leftarrow T_i + 1$.
\item At the end of $K$ rounds, i.e.\ when $T_i = K$ (for some $K$ to be determined later in our analysis), if $F_i = 1$, commit to $b_i^*$ and return $b_i^*$ as
output bit.
\item After completing the above protocol, start the epoch synchronization process again.
\item If after waiting $5\delta$ time, no further step can be taken during any point of the above protocol, start with the next iteration at the beginning of this loop.
\end{enumerate}
\end{enumerate}
\end{mdframed}
\caption{Linear multicast Byzantine Agreement protocol.}\label{alg:linear-bba-alg}
\end{figure}

\subsection{Warmup Protocol Analysis}\label{sec:warmup-protocol-analysis}

We define the round complexity of our protocol to be rounds of communication during the synchronous period after GST.
In our analysis, we assume the following two parameters.
Let $T$ be the number of rounds during the asynchronous period of time
before GST. First, we assume that $T = O\left(3^{\log^c{n}}\right)$
for \emph{any} constant $c$. Secondly, we assume that
the synchronous period that immediately follows this asynchronous period lasts for $\Omega\left(\log T\right)$ rounds.\footnote{This assumption is similar to the
assumption assumed by~\cite{algorand}.}

Our protocol stated above uses $\Omega(n)$ multicast messages. This is not
ideal and we fix this assumption in~\cref{sec:sublinear-bba}
under a somewhat different partially synchronous model.
But first, we prove the correctness and round complexity (and hence multicast complexity)
of our simple protocol given above in the partially synchronous model with an adaptive adversary.

\paragraph{Round Synchronization} We first show that after GST, the pre-confirmed
epoch number always progresses forward (i.e. a replica $i$ would never reset
$C_i$ to a lower epoch number). See~\cref{alg:linear-clock-alg} for the epoch synchronization protocol.

\begin{lemma}\label{lem:confirmation-round-increment}
After GST, a replica $i$'s $C_i$ increments by at least $1$.
\end{lemma}

\begin{proof}
Suppose a replica $i$'s $C_i = j$. After GST, $i$ receives all messages sent by honest replicas. Then, one of two events can happen:

\begin{enumerate}
\item An honest replica $k$ has a $C_k$ where $C_k > j$ (and replica $i$ receives a certificate for $C_k$).
Then, $i$ will update $C_i \leftarrow C_k$ and $i$ has incremented its $C_i$ by at least $1$.

\item $j \geq C_k$ for any other honest replica $k$. Then, by our epoch synchronization
protocol given in~\cref{alg:linear-clock-alg}, $i$ proposes $j + 1$ as
the new round. Since all other $C_k \leq j$, all $k$ will vote for $j+1$. Thus, $C_i$ increments by at least $1$.

Thus, by induction, this is true for all values of $C_i$, provided the base case when $C_i$ is set to $0$ initially.
\end{enumerate}
\end{proof}

We now show that an adversary cannot selectively choose an arbitrarily large epoch number
to synchronize the clock with high probability.

\begin{lemma}\label{lem:num-rounds-after-gst}
Let the largest pre-confirmed epoch immediately before GST held by any honest replica in round $t$ of the
asynchronous period be $C$. Then, assuming that the GST ending round $t$ lasts $\Theta(\log T)$ rounds, then the
largest pre-confirmed epoch held by any honest replica at the end of GST is $O(C + \log T)$.
\end{lemma}

\begin{proof}
The replica holding the largest pre-confirmed epoch $C$ will propose $C + 1$ after GST. All honest replicas will subsequently vote for $C +1$ and the largest pre-confirmed
round becomes $C + 1$. We now prove that during each epoch
after GST, the epoch number increases by exactly $1$.
Suppose for contradiction that this is not the case. Suppose that $C_i > C$ for replica $i$ increases from $C_i$ to $C_i + 2$
after one round of communication after GST. Then, $i$ must not have heard $C_i + 1$ being proposed
in this epoch or a previous epoch.
Since $i$ will increase $C_i$ to $C_i + 2$, then either it never received a
proposal for $C_i + 1$ or received a pre-commit certificate for $C_i + 1$.
Either way, this implies that $i$ has not received a proposal for $C_i + 1$
in the current epoch or a previous epoch. Since the network becomes
synchronous after GST, this is impossible. Thus, all honest replicas $i$ increase their $C_i$ by at most $1$ each round after GST.

Therefore, since the total number of rounds after GST is $\Theta(\log T)$,
the largest confirmed epoch number held by any honest replica at the
end of GST will be $O(C + \log T)$.
\end{proof}

\begin{lemma}
Let the largest pre-confirmed epoch at the start of the protocol held by any honest replica be $C$. Assuming that the period of asynchrony lasts $T = O\left(3^{\log^c n}\right)$ rounds
and GST lasts $\Theta(\log T)$ rounds, the largest pre-confirmed epoch held by any honest replica at the start of the first round after GST is $O\left(C + T + \log T\right)$.
\end{lemma}

\begin{proof}
The largest pre-confirmed epoch cannot increase by more than $1$ following any communication round. Let the largest pre-confirmed epoch held by any replica
at round $j \leq T$ be $K$. Then, no replica at the start of round $j + 1$ holds a certificate to pre-confirmed epoch $K + 1$. Thus, no replica will vote
for round $K + 2$ since all replicas will propose round $\leq K + 1$. Hence, the largest pre-confirmed epoch cannot increase by more than $1$ following
any communication round even during periods of asynchrony.
Since the period of asynchrony lasts $T$ rounds, at the start of GST, the largest pre-confirmed epoch held by any honest replica will be $O(C + T)$.
By~\cref{lem:num-rounds-after-gst}, the largest pre-confirmed epoch
held by any honest replica after GST will be $O\left(C + T + \log T\right)$.
\end{proof}

The way the adversary can take advantage of epoch synchronization for
their benefit is to synchronize to an epoch where a large number of Byzantine replicas
become leaders before the necessary number of rounds of honest leaders occur.
In such cases, honest replicas could be prevented from reaching consensus for
longer than polylogarithmic number of rounds. We first present here a lemma
that we later use in our argument that such a case most likely will not occur with
high probability.

\begin{lemma}\label{lem:consecutive-rounds}
Let the period of asynchrony be $O(\async)$ rounds. Then, the probability that $X$ consecutive rounds
in this period of asynchrony have at least one Byzantine leader during each of the $X$ rounds
is upper bounded by $T \left(2fD_0\right)^X$ where $D_0$ is defined in~\cref{alg:linear-bba-alg}.
\end{lemma}

\begin{proof}
Given that there are at most $f$ Byzantine replicas (where there are at least $n \geq 3f + 1$ total replicas)
at any point in time and each has two chances to become a leader
for any round (once for bit $0$ and another chance for bit $1$), the probability that a Byzantine replica proposes
a bit in any round is at most $2fD_0$ where $D_0 \leq 1$ is our threshold for proposing a bit.
Given a particular starting round, the probability that the next $X$ consecutive rounds have at least one
Byzantine replica proposing a bit is $\left(2fD_0\right)^X$. By the union bound over all
possible starting rounds in the period of asynchrony, the probability that there exists such a sequence of $X$ rounds is upper
bounded by $T\left(2fD_0\right)^X$.
\end{proof}

Using the above we show that if the period after GST is sufficiently long,
then we synchronize to a round with no Byzantine leaders with
high probability.

\begin{lemma}\label{lem:termination-time}
Let the period of asynchrony be $O(T)$ rounds. Then, with high probability, the protocol reaches a
round with no Byzantine leaders in $O\left(\log_{\frac{1}{2fD_0}}\left(nT\right)\right)$ rounds after GST (assuming $f > 0$).
\end{lemma}

\begin{proof}
By~\cref{lem:consecutive-rounds}, the probability that there exists a
period of $X$ rounds where at least one adversarial leader exists
in each of the $X$ rounds is $T\left(2fD_0\right)^X$. We want the smallest $X$ where $T\left(2fD_0\right)^X \leq \frac{1}{n^c}$ for all constant $c$
by which we obtain with high probability an epoch with no Byzantine
leaders after $X$ epoch after GST (since the adversarial strategy is
to synchronize the epoch to the beginning of the $X$ epochs after GST).
Solving for $X$, we obtain, $X \geq c\log_{1/2fD_0}\left(n\right) + \log_{1/2fD_0}\left(T\right)$.
\end{proof}

Now we show that we synchronize to a round with exactly one honest leader.

\begin{lemma}\label{lem:one-honest}
Let $D_0 = 1/2n$ and the period of asynchrony be $O(T)$ rounds. Then, with high probability, the protocol reaches a
round with exactly one honest leader in $O\left(\log \left(nT\right) + \log n\right)$ rounds after GST.
\end{lemma}

\begin{proof}
By~\cref{lem:termination-time}, we obtain a round with no Byzantine leaders
with high probability in $O\left(\log \left(nT\right)\right)$
rounds after GST if we assume $D_0 = 1/2n$. By the Chernoff bound, with high probability, exactly one honest leader will exist after $O(\log n)$ rounds.
Thus, exactly one honest leader exists in a round after GST in $O\left(\log \left(nT\right) + \log n\right)$ rounds.
\end{proof}

\begin{lemma}\label{lem:honest-round-sublinear}
Round synchronization using sublinear number of rounds to a round with exactly one honest leader occurs with high probability after GST
if we assume the period of asynchrony lasts $O\left(3^{n^{\epsilon}}\right)$ rounds assuming $D_0 = 1/2n$ and for all $0 < \epsilon < 1$.
\end{lemma}

\begin{proof}
By~\cref{lem:one-honest}, the number of rounds required before reaching a round with exactly one honest leader is $O\left(\log \left(nT\right) + \log n\right)$.
Setting $T = 3^{n^{\epsilon}}$, we obtain that the number of rounds required
before reaching an epoch with exactly one honest leader is $O\left(n^{\epsilon} + \log n\right)$,
which is sublinear if we assume $0 < \epsilon < 1$.
\end{proof}

\begin{corollary}\label{cor:polylog-rounds}
Epoch synchronization using polylogarithmic number of rounds to an
epoch with exactly one honest leader occurs with high probability after GST
if we assume the period of asynchrony lasts $O\left(3^{\log^c n}\right)$ rounds assuming $D_0 = 1/2n$ and for all constants $c \geq 0$.
\end{corollary}

\paragraph{Consistency within an epoch.} In this section, we show that if two replicas
receive sufficiently many votes for their bits $b_i^*$ and $b_j^*$, then $b_i^* = b_j^*$.

\begin{lemma}\label{lem:consistency}
Suppose a forever honest replica $i$ observed $\geq \floor{2n/3}$ $\ack$s from a set of
replicas $S$ for $b_i^*$ and forever honest
replica $j$ also receives $\geq \floor{2n/3}$ $\ack$s from replica set $S'$ for $b_j^*$. Then, at least one forever honest replica
exists in $S \cap S'$ and $b_i^* = b_j^*$.
\end{lemma}

\begin{proof}
We prove this via contradiction. Suppose that honest replica $i$ observed $\geq \floor{2n/3}$ $\ack$s from a set of
replicas $S$ for $b_i^*$ and forever honest
replica $j$ also receives $\geq \floor{2n/3}$ $\ack$s from replica set $S'$ for $b_j^*$ but $b_i^* \neq b_j^*$.
By the assumptions of our network model, this means that at least $\floor{2n/3} + 1$ replicas voted for $b_i^*$
and a disjoint set of at least $\floor{n/3} + 1$ replicas voted for $b_j^* \neq b_i^*$ since honest replicas vote at most once per round and
there exists at most $\floor{n/3}$ Byzantine replicas in the network.
Thus, the total number of replicas in the network must be $> n$, a contradiction.
Therefore, by the pigeonhole principle, there exists at least one
forever-honest replica in the set $S \cap S'$.
\end{proof}

\paragraph{Termination in $O\left(\log\left(nT\right) + \log n\right)$ rounds after GST with high probability.}
After GST, the various $C_i$ and $C_j$ for honest replicas $i$ and $j$ might be de-synced. Suppose
$D_0 = \frac{1}{2n}$, we show that with high probability,
there exists an epoch $R_i$ where some honest replicas see $\geq \floor{2n/3}$ $\ack$s
on a single bit after $O\left(\log \left(nT\right) + \log n\right)$ rounds after GST assuming the period of asynchrony lasted $O(T)$ rounds. If all honest replicas
start with the same input bit $b$, then, with high probability, after $O\left(\log \left(nT\right) + \log n\right)$ rounds, all replicas will output $b$.
It takes $O\left(\frac{\Delta\cdot D_0}{\delta}\right)$ proposal messages
before all replicas are synced to the same epoch, the smallest epoch that has not been pre-committed.
Let this round be $R_k$. Setting $\delta$ to be on the
order of $\Delta$, we obtain $O(1)$ rounds after GST after which the
replicas will be synced to the same epoch. This is proven in the following lemma:

\begin{lemma}\label{lem:replica-syncing}
If $\delta = \Theta(\Delta)$, after $O(1)$ rounds after GST,
all replicas will be synced to the same epoch with $O(n)$ multicasts
using the distributed clock protocol presented in~\cref{alg:linear-clock-alg}.
\end{lemma}

\begin{proof}
After GST, let $\ell$ be the smallest round for which a replica sends a proposal that has not been
pre-confirmed. After $\Delta$ delay, all replicas receive the message with round $\ell$ attached.
By definition of $\ell$, there does not exist a proposal that has not been confirmed that is smaller than
$\ell$. After $\Delta$ delay, all replicas receive the proposal $\ell$. Since no honest replicas have pre-confirmed $\ell$,
all honest replicas will vote for $\ell$ until $\ell$ is confirmed. Because there are at most $3$ rounds of voting and one round
of proposal and during each of these communication rounds each honest replica only sends one message,
if $\delta = \Theta(\Delta)$, after $O(1)$ rounds after GST, all replicas will be synced to the same round
with multicast complexity $O(n)$.
\end{proof}

\begin{lemma}\label{lem:termination}
Suppose $\delta \approx \Delta$ and $D_0 = 1/2n$. All honest replicas will multicast $\ack$ for at least one honest proposal in $O\left(\poly \log n\right)$ rounds after GST with high probability.
\end{lemma}

\begin{proof}
Let $\ell$ be the first round where (a) there exists exactly one leader and
(b) the one leader is an honest replica.
The probability that this occurs in any particular round is
$\geq \left(\frac{2n}{3}\right)\left(\frac{1}{2n}\right)\left(1-\frac{1}{2n}\right)^{n-1} \geq \left(\frac{1}{3}\right)\left(\frac{2n}{\sqrt{e}(2n-1)}\right) \geq \frac{1}{6}$.
Then, after GST, $\ell$ occurs at least once
in $O(\log n)$ rounds (after all replicas are synced)
with probability at least $\geq 1 - \left(\frac{5}{6}\right)^{c\log n} \geq 1-\frac{1}{n^c}$ for any constant $c > 0$.
Since all replicas are now synced and there is only one leader and it is honest,
all honest replicas will $\ack$ the honest proposal. Suppose $T = O\left(3^{\log^c n}\right)$
for any constant $c > 0$, by~\cref{lem:one-honest} and what we showed above,
a total of $O\left(\poly \log n\right)$ rounds are necessary for all
honest replicas to multicast $\ack$ for at least one honest proposal after GST with high probability.
\end{proof}

\paragraph{A good epoch exists in $O(\log n)$ rounds where only one honest leader proposes a bit with high probability.}
We define \emph{good epochs} similarly to the definition given in~\cite{cc19} except
for the partially synchronous model. Let $h$ be an honest leader in $\epoch{i}$. Then, given that
$h$ passes the test for sending a proposal (and since the probability of picking a bit and of sending a proposal are independent), $h$ chooses a \emph{lucky bit}
$b^*$ (uniformly at random) iff either 1) in $\epoch{i} - 1$, no honest replicas have seen $\frac{2n}{3}$ $\ack$s for its own bit; or 2) in $\epoch{i} - 1$,
some honest replicas have seen $\frac{2n}{3}$ $\ack$s for its own bit and this bit equals $b^*$.
The honest leader $h$ chooses a lucky bit $b^*$ with probability at least $1/2$.
Thus in $O(\log n)$ epochs where only one honest leader proposes a bit, a good epoch exists with $1 - \frac{1}{n^c}$ probability for all constant $c > 0$.

\paragraph{Persistence of honest choice after a good epoch.} Once we have a good epoch $\epoch{i}$
as defined above, all honest replicas will $\ack$ $b^*$ during $\epoch{i}$. Once all honest replicas have the same
$b^*_i$, then they will never set their $F_i = 0$ again. Thus, they will $\ack$ $b^*$ for all future rounds.
By induction, in all future epochs they will stick to $\ack$ing $b^*$.

\paragraph{Validity} If all honest replicas receive the same bit $b^*$ as input then no honest replica will
set $F_i = 0$ and $b^*$ will be output by all honest replicas.

Using the above, we obtain the following theorem.

\begin{theorem}
Binary Byzantine agreement can be reached in $O(\poly \log n)$ rounds after GST with high probability
and with $O(n)$ multicast complexity assuming the
asynchronous period lasts for no more than $O\left(3^{\log^c n}\right)$ rounds.
\end{theorem}

\begin{proof}
By~\cref{lem:termination-time},~\cref{cor:polylog-rounds},
and~\cref{lem:replica-syncing}, we can achieve a round
with one leader who is honest after $O(\poly \log n)$ rounds
after GST with high probability.
Then, by~\cref{lem:termination},
after $O(1)$ additional communication rounds, we achieve consensus on one bit.
\end{proof}
\section{Constructions of $\cryptosort$}

Our protocols in~\cref{sec:general-consensus} and~\cref{sec:sublinear-bba} require that $\cryptosort$ be a
signing oracle that has the following properties:

\begin{enumerate}
    \item With all but negligible probability in $\secp$, no other replica $j \neq i$ can
        predict the output of $\cryptosort(\sk{i}, x)$ on the secret key for replica $i$
        and input $x$ without querying $\cryptosort$.
    \item An output from $\cryptosort$ can be verified by all other replicas
        (the output from $\cryptosort$ may include an additional proof).
    \item The output of $\cryptosort$ is a value generated uniformly at random in $[0, 1]$.
\end{enumerate}

There are several instantiations of such functions. One of these instantiations is
via verifiable random functions (VRFs) as described in~\cref{app:vrfs}.
Another instantiation is the nice real-world cryptographic
construction presented in Appendix D of~\cite{cc19}
using PRFs and adaptively-secure NIZKs.

\subsection{Verifiable Random Functions (VRFs)}\label{app:vrfs}
In this section, we provide the full, formal definition of VRFs as provided in~\cite{vrf} and then a brief description
of how to use such a function for our cryptographic sortition signature oracle $\cryptosort$.

\begin{definition}[Verifiable Random Functions (VRFs)~\cite{vrf}]\label{def:vrf}
    Let $(G, F, V)$ be a set of three polytime algorithms. The \emph{function generator}
    $G\left(1^{\secp}\right) \rightarrow (\pk{i}, \sk{i})$ outputs a
    public key/private key pair $(\pk{i}, \sk{i})$ which are two binary strings. The \emph{function evaluator}
    $F=(F_1, F_2)$ is a two-part function
    that each takes as input the secret key $\sk{i}$ and an input $x$
    and outputs a value $\vout{} \xleftarrow{F_1}(\sk{i}, x)$ and
    a proof $\proofout{} \xleftarrow{F_2}(\sk{i}, x)$.
    The \emph{function verifier} $V$ takes as input
    $\{Yes, No\} \xleftarrow V(\pk{i}, x, \vout{}, \proofout{})$
    and outputs $Yes$ or $No$.

    Let $a: \mathbb{N} \rightarrow \mathbb{N} \cup \{*\}$ and $b, s: \mathbb{N} \rightarrow \mathbb{N}$
    be any three functions such that $a(\secp)$, $b(\secp)$, $s(\secp)$ are all computable in time $poly(\secp)$
    and $a(\secp)$ and $b(\secp)$ are both bounded by a polynomial in $\secp$ (except when $a(\secp)$ takes on the value
    $*$ which means the set of all binary values).

    $(G, F, V)$ is a \emph{verifiable pseudorandom function (VRF)}
    with input length $a(\secp)$, output length $b(\secp)$, and security $s(\secp)$ where $\secp$ is the security parameter
    if the following properties hold:

    \begin{enumerate}
        \item \textbf{Domain-Range Correctness and Provability}:
            The following conditions hold for all but negligible probability in $\secp$ ($negl(\secp)$):
            \begin{enumerate}
                \item For all $x \in \{0, 1\}^{a(\secp)}, F_1(\sk{i}, x) \in \{0, 1\}^{b(\secp)}$.
                \item For all $x \in \{0, 1\}^{a(\secp)}$, if $(\vout{}, \proofout{}) \xleftarrow F(\sk{i}, x)$,
                    then $\prob[V(\pk{i}, x, \vout{}, \proofout{}) = Yes] > 1 - negl(\secp)$.
            \end{enumerate}
        \item \textbf{Unique Provability}:
            For every $\pk{i}, x, \vout{1}, \vout{2}, \proofout{1}, \proofout{2}$ such that $\vout{1} \neq \vout{2}$,
            the following holds for either $g = 1$ or $g = 2$:
            \begin{align*}
                \prob[V(\pk{i}, x, \vout{g}, \proofout{g}) = Yes] < negl(\secp).
            \end{align*}
        \item \textbf{Pseudorandomness}: Given a probabilistic polynomial time (PPT) algorithm
            $\alg = (\alg_E, \alg_T)$ that both run for at most $s(\secp)$ steps when their first inputs
            are $1^\secp$ and \emph{does not query the oracle $F(\sk{i}, \cdot)$ for $x$}, then $\alg$ succeeds
            in the following experiment with probability $1/2 + negl(\secp)$:
            \begin{enumerate}
                \item Run $G(1^\secp)$ to obtain $(\pk{i}, \sk{i})$.
                \item Run $\alg_E^{F(\sk{i}, \cdot)}(1^\secp,\pk{i}) \rightarrow (x, state)$.
                \item Choose $b \xleftarrow{R} \{0, 1\}$.
                    \begin{enumerate}
                        \item If $b = 0$, let $\vout{} \xleftarrow{F_1}(\sk{i},x)$.
                        \item If $b = 1$, let $\vout{} \xleftarrow{R} \{0, 1\}^{b(\secp)}$.
                    \end{enumerate}
                \item Run $\alg_T^{F(\sk{}, \cdot)}\left(1^\secp, \vout{}, state\right)$ to obtain $guess$.
                \item $\alg = (\alg_E, \alg_T)$ succeeds if $x \in \{0, 1\}^{a(k)}$, $guess = b$, and
                    $\alg$ did not query $F(\sk{i}, x)$.
            \end{enumerate}
    \end{enumerate}
\end{definition}

Intuitively, VRFs are functions which takes as input a secret key $\sk{i}$ and some arbitrary function input $x$
and returns an output $\vout{x}$ and a proof $\proofout{x}$. Using the input $x$, output $\vout{x}$, proof $\proofout{x}$,
and public key $\pk{i}$, anyone can verify quickly that the VRF was computed (using the player's secret key).
Furthermore, an adversary cannot guess the
output of the VRF without computing it except with negligible probability in our security parameter $\secp$. This means
that an adversary cannot guess (except with negligible probability) the output of anyone's VRF without knowing their
secret key.

There are many instantiations of verifiable random functions (some from standard assumptions) (e.g.\ \cite{Dodis02,DY05,HJ16}).
The proof sizes of many of these constructions are $O(\secp)$. Thus,
the cost in bits of sending these proofs in messages is essentially the same as that necessary
to send a signature using a standard PKI.

\section{Verifiable Delay Functions (VDFs)}\label{app:vdf}
The formal definition of VDFs is presented below.

\begin{definition}[Verifiable Delay Functions~\cite{vdf}]\label{def:vdfs}
    A VDF $V = \left(\setup, \eval, \veri\right)$ is a triple of algorithms that perform the following:
    \begin{enumerate}
        \item $\setup(\secp, \diff{}) \rightarrow \textbf{pp} = (\ek{}, \vk{})$: The $\setup$ algorithm takes
            as input a security parameter $\secp$ and a desired difficulty level $\diff{}$ and produces
            public parameters consisting of an evaluation key $\ek{}$ and a verification key $\vk{}$.
            $\setup$ is polynomial time with respect to $\secp$ and $\diff{}$ is subexponentially-sized
            in terms of $\secp$. The public parameters specify an input space $\sX$ and an output space $\sY$.
            $\sX$ is efficiently sampleable. If secret randomness is used in $\setup$, a trusted setup might
            be necessary.
        \item $\eval(\ek{}, x) \rightarrow (y, \pi)$: $\eval$ takes an input $x \in \mathcal{X}$ (in the
            sample space of inputs) and the evaluation key and produces an output $y \in \mathcal{Y}$ (in the
            sample space of outputs) and a (possibly empty) proof $\pi$. $\eval$ may use random
            bits to generate $\pi$ but not to compute $y$. $\eval$ runs in parallel time $\diff{}$
            even when given $poly(\log(\diff{}), \secp)$ processors for all $\textbf{pp}$ generated
            by $\setup(\secp, \diff{})$ and $x \in \sX$.
            \item $\veri(\vk{}, x, y, \pi) \rightarrow \{Yes, No\}$: $\veri$ is a deterministic algorithm
                that takes the verification key $\vk{}$, an input $x$, the output $y$, and proof $\pi$ and outputs
                $Yes$ or $No$ depending on whether $y$ was correctly computed from via $\eval$. $\eval$ runs in
                time $poly(\log(\diff{}),\secp)$.
    \end{enumerate}
    Furthermore, $V$ must satisfy the following properties:
    \begin{enumerate}
        \item \textbf{Correctness} A VDF $V$ is correct if for all $\secp, \diff{}$,
            parameters $(\ek{}, \vk{}) \xleftarrow{R} \setup(\secp, \diff{})$,
            and all $x \in \sX$, if $(y, \pi) \xleftarrow{R} \eval(\ek{}, x)$,
            then $\veri(\vk{}, x, y, \pi) \rightarrow Yes$.
        \item \textbf{Soundness} A VDF is sound if for all algorithms $\alg$ that run in time
            $O\left(poly(\diff{}, \secp)\right)$
            \begin{align*}
                \prob\left[
                    \begin{array}{l}
                        \veri(\vk{}, x, y, \pi) = Yes \\
                        y' \neq y
                    \end{array}
                \middle\vert
                    \begin{array}{l}
                        \mathbf{pp} = (\ek{}, \vk{}) \xleftarrow{R} \setup(\secp, \diff{})\\
                        (x, y', \pi') \xleftarrow{R} \alg\left(\secp, \mathbf{pp}, \diff{}\right)\\
                        (y, \pi) \xleftarrow{R} \eval(\ek{}, x)
                    \end{array}
                \right] \leq negl(\secp).
            \end{align*}
        \item \textbf{Sequentiality} A VDF is $(p, \sigma)$-sequential if no adversary $\alg = (\alg_0, \alg_1)$
            with a pair of randomized algorithms $\alg_0$,
            which runs in total time $O(poly(\diff{}, \secp))$, and $\alg_1$, which runs in parallel time $\sigma(t)$
            on at most $p(t)$ processors, can win the following game with probability greater than $negl(\secp)$:
            \begin{flalign*}
                &\textbf{pp} \xleftarrow{R} \setup(\secp, \diff{}) \\
                &L \xleftarrow{R} \alg_0(\secp, \textbf{pp}, \diff{}) \\
                &x \xleftarrow{R} \mathcal{X} \\
                &y_{\alg} \xleftarrow{R} \alg_1(L, \textbf{pp}, x).
            \end{flalign*}
            $\alg = (\alg_0, \alg_1)$ wins the game if $(y, \pi) \xleftarrow{R} \eval(\ek{}, x)$ and $y_{\alg} = y$.
    \end{enumerate}
\end{definition}

For the VDFs used in our constructions, we use more efficient VDFs that provide tighter time
bounds than the bounds given by the definition of VDFs above. More specifically, we consider
VDF constructions where the definition of sequentiality $p$ is for any polynomial and
$\sigma(t) = t-\epsilon t$ for some sufficiently small constant (but we even prove our protocol
secure for any constant gap between the speeds of the adversary and the honest party).

There are several well-known VDF constructions in the literature~\cite{vdf,Pietr19,Wes19}, most of which
use exponentiation but some use other methods~\cite{BDRV19,EFKP19,FMPS19}.
Because some of these constructions uses public-coin setup and public-coin
succinct arguments for proving the correctness of the output, one can use the Fiat-Shamir Heuristic
to make the proofs non-interactive; hence, the VDFs used in our protocol in~\cref{sec:general-consensus}
include such proofs as part of the output. The proof sizes of these constructions are generally $O(\poly(\secp))$
or even, $O(\log D)$ (see e.g. \cite{Pietr19,Wes19,FMPS19}).

\iflong
\section{Additional Background}\label{app:additional}

Consensus protocols have been studied since the 1980s as a method to
provide fault tolerance for information stored in
databases~\cite{dls88,pbft}. Such protocols achieve fault tolerance by
distributing data and computation across many different machines which
are physically separated from one another over a network; they aim to
achieve \emph{safety}, where participants agree on the same
committed values and each has a copy of the same history of committed values,
and \emph{liveness}, where progress is continuously being made to increase the log of
committed values.
The state machine replication (SMR) model is a type of consensus model
where many different copies of a machine, called replicas, run the
same protocol and communicate with one another via communication
channels to create a common log. Such protocols are often formulated in
models where there exist both a communication delay and
desynchronized clocks~\cite{dls88}; such models accurately
depict certain scenarios in the real-world. Throughout this paper, we refer
to obtaining consensus in the SMR model, simply, as \emph{consensus}
and consensus protocols in the SMR model as \emph{consensus
  protocols}.

\subsection{Network Models}

Consensus protocols are often studied in various models that are
designed to reflect network conditions in the real-world.

Some network models~\cite{dls88} consider cases where in addition to
message delay, replicas can have internal clocks that run at different
speeds, causing the clocks held by replicas to be
desynchronized. Usually, we are concerned with the case where network
clocks are desynchronized by at most $\Phi$, where all processors take
at least one step in any size $\Phi$ block of time. Although there are
standard techniques for synchronizing the clocks held by
processors~\cite{dls88},
in cases of unknown message delay, such techniques require $\Omega(n^2)$ message complexity (or $\Omega(n)$ multicast complexity)
which is too large for the settings we consider in this paper. As is the case with more
recent research on communication-efficient protocols secure against adaptive adversaries~\cite{cc19,CPS19,algorand},
we will only consider communication delay/asynchrony in this paper without considering processor delay
or processor asynchrony as real-world networks often exhibit synchronized global clocks.
In the case of our impossibility results, since we consider a stronger communication model,
our results also hold for the case when processors exhibit different speeds bounded by $\Phi$.
In the future, we intend to extend our protocols to also account for the case when processor
speeds also differ by some bounded $\Phi$.

Abstracting the conditions of the real-world gives us three main network models:

\paragraph{Synchronous} The synchronous model is a model where \emph{all messages}
sent by honest nodes are expected to arrive (at honest nodes) after $\Delta$ delay.

\paragraph{Asynchronous} The asynchronous model is a model where the message
delay is unbounded. This means that messages can be dropped at any point in the protocol.

\paragraph{Partially Synchronous} The partially synchronous model is a model where
there exists an unknown (but bounded) period of time during which the network is asynchronous.
After Global Stabilization Time (GST), the network becomes synchronous with network message delay
$\Delta$.

Although there are other network models, we do not consider them in this paper and, thus
do not expand upon them here.  In \cref{sec:sublinear-bba} we introduce a new
network model in which the adversary can drop messages with some fixed probability $p$.

\subsection{Adversarial Strategies and Advantages}\label{sec:adv-advantages}

An adaptive adversary can perform a number of actions to maximize
their advantage over the honest replicas which correctly follow the
prescribed protocol.
For example, if the leader election schedule is predictable, the
adversary can immediately corrupt the leader before
they are elected, reducing chain quality or censoring proposals.
In committee election systems,
\footnote{
A \emph{committee election system}  is one where voting and block proposals are performed only by members of  a committee which is elected uniformly at random.}
the committee is typically much smaller than the tolerated number of
faults.  Once the adversary sees the committee they can immediately
corrupt the entire committee. Previous work solved this problem via
the \emph{memory-erasure} model~\cite{algorand}. However, erasure is hard
to perform in real-life systems, and furthermore, without a
proof-of-erasure it is not possible to check whether such an erasure
was actually performed. Thus, it is ideal to remove such an assumption
when protecting against this form of attack.  We call this type of
attack \emph{key reuse}:

\paragraph{Key Reuse}
Without the memory-erasure model, an adversary can sign multiple messages as the replica once it is corrupted.
In the adaptive adversary setting with leader election, the adversary
can wait to see who sends messages as the leader and then immediately
corrupt that replica to send out many different proposals, splitting
votes and preventing consensus. In protocols which elect committees,
the adversary can similarly wait to see who sends messages as a
committee member and then immediately corrupt those replicas to vote
for multiple proposals, violating safety.

One way to address this is vote-specific eligibility---election is
tied to the proposed value, so even a corrupted replica cannot vote
for two values.  In these protocols, the adversary can use
computational power to get unfairly elected to the committee.  In the
case of public ledgers, honest replicas will only propose proposals
sent to them by clients. This means that honest replicas, which gossip
only the proposals they receive, do not have a disproportionate
chance to become part of the committee.  However, the adversary can
create and try many different arbitrary proposals, for example spending coins back to itself, to increase the
chances that Byzantine replicas are elected to committees. We formally
define this process as \emph{vote grinding}:

\paragraph{Vote Grinding} In previous work~\cite{cc19,CPS19}, the election of a replica to be leader
or to a voting committee
is dependent probabilitistically on the specific proposal/vote proposed by the replica (chosen with uniform probability over all
proposals).
An adversary can choose to try many proposals
in parallel to attempt to obtain a large number of Byzantine replicas in the committee. Even though the probability
of a replica being elected to a committee is small, the chances of a set of Byzantine replicas
of sufficient size being part of the
committee approaches $1$ as the number of proposals tried increases.
Since the proposed adversarial proposals
are not tied to real proposals requested by clients, the adversary can try a much larger set of proposals than
the total number of proposals requested by clients in the network. We call this form of attack \emph{vote grinding}.

In our paper, we protect against both attacks above without the use of the memory-erasure model in the synchronous model.
We hope that such techniques could be applied to the partially-synchronous setting as studied in previous works~\cite{algorand}.
However, one must limit the number of rounds the asynchronous phase can last due to the following possible attack:

\paragraph{Fast-forwarding} In protocols with leader and committee election (assuming committee size is $3\lambda + 1$ for some $\lambda$), there is some execution where, given an adversary that has corrupted $f$ replicas, the leader chosen will happen to be one of the $f$ or the committee chosen will consist of too many adversarial members.
The adversary can try to set the clock to a round number where a long sequence of proposers and
committee members are all adversarial (i.e. they are dominated by the $f$ nodes currently controlled by the
adversary).
The probability that any $X$ consecutive rounds has $>\frac{\lambda}{3}$ adversaries in the committee from
a given set of $f$ adversarial nodes is $\geq \left(\frac{f}{n}\right)^{\frac{(\lambda + 3)X}{3}} > 0$
when $X$ is finite (recall that the committees consist of $3\lambda + 1$ members now instead of $3f + 1$). Thus, by the probabilistic
method, such a set of $X$ consecutive rounds is guaranteed to hold for some finite consecutive round numbers.
Because such a long sequence is unlikely to occur in polylogarithmic number of rounds,
in the synchronous model we can still achieve consensus with high probability.
However, in the partially synchronous model, the adversary can drop or
delay messages arbitrarily during the asynchronous period until GST.
If we assume that the adversary has computing power bounded by a large enough polynomial,
we should also assume they are able to
find this consecutive sequence.
During a long enough asynchronous period the adversary
can drop messages and force repeated elections until a leader and committee
are elected where it has this specific advantage.  At this point the
adversary can propose and vote for multiple proposals during multiple consecutive rounds,
eventually violating safety.

\paragraph{Partitioning committees} During the asynchronous parts of the partially synchronous
model, the adversary has complete control of the network, e.g.\ when messages are sent and received. This means
that it is possible that some nodes might believe that other nodes have agreed on a bit while others
believe they have agreed on a different bit. In fact, we show that is it impossible to know when
all nodes have agreed on the same bit in BBA using sublinear multicasts in this model in~\cref{sec:impossible-bba}.

\fi
\end{document}